\documentclass[11pt]{llncs}
\usepackage{qip}
\usepackage{framed,xspace}
\usepackage{fullpage}
\usepackage{graphicx}
\newcommand{\rootprob}{\frac{1}{\sqrt{|\mathcal{R}|}}}
\newcommand{\prob}{\frac{1}{{|\mathcal{R}|}}}

\begin{document}
\title{Superposition Attacks on Cryptographic Protocols}
\author{Ivan Damg{\aa}rd, Jakob Funder, Jesper Buus Nielsen, Louis Salvail}
\institute{Dept. of Computer Science, Aarhus University, Universit\'{e} de Montreal}
\date{}
\maketitle
\begin{abstract}
Attacks on classical cryptographic protocols are usually modeled by allowing an adversary to ask queries from an oracle. Security is then defined by requiring that as long as the queries satisfy some constraint, there is some problem the adversary cannot solve, such as compute a certain piece of information.
In this paper, we introduce a fundamentally new model of quantum attacks on classical cryptographic protocols, where the adversary is allowed to ask several classical queries in quantum superposition. This is a strictly stronger attack than the standard one, and we consider the security of several primitives in this model. We show that a secret-sharing scheme that is secure with threshold $t$ in the standard model is secure against superposition attacks if and only if the threshold is lowered to $t/2$. We use this result to give zero-knowledge proofs for all of NP in the common reference string model. While our protocol is classical, it is sound against a cheating unbounded quantum prover and computational zero-knowledge even if the verifier is allowed a superposition attack. Finally, we consider multiparty computation and show that for the most general type of attack, simulation based security is not possible. However, putting a natural constraint on the adversary, we show a non-trivial example of a protocol that can indeed be simulated.
\end{abstract}

\section{Introduction}
Attacks on classical cryptographic protocols are usually modeled by allowing an adversary to ask queries from an oracle, for instance the adversary specifies a subset of parties he wants to corrupt, and gets back their views of the protocol. Security is then defined by requiring that as long as the queries satisfy some constraint (for instance, the corrupted subset is not too large), there is some problem the adversary cannot solve, such as compute a certain piece of information.

Several previous works consider what happens to security if we allow the adversary to be quantum. The model usually considered is that the adversary is now a quantum machine, but otherwise plays exactly the same game as in a classical attack, i.e., he still communicates classically with the protocol he attacks. One example of this is the work of Watrous, showing that a large class of zero-knowledge protocols are also zero-knowledge against a quantum verifier.

It is natural to ask why we constrain a quantum adversary to communicate classically during the attack? The standard answer to this is that since honest players are classical, they would (implicitly) be doing a measurement of anything they receive, thus forcing a collapse of any quantum state they are given. 

An important point, however, is that the assumption about honest players being classical is not always justified, {\em even if the protocol is supposed to be classical}: in the future, honest players may use quantum computing, just to speed up their local computation, even if they sometimes communicate classically. Furthermore, future usage of quantum cryptography will imply that players sometimes
communicate quantumly (to do quantum key distribution) and sometimes classically. Finally, one should consider the case where a classical protocol is used as a subrutine for a protocol that handles quantum data. This is exactly what happens in the work of Ben-Or et al.~\cite{BCGHS06}, where classical multiparty computation is used as a tool to obtain quantum multiparty computation. 

Now, if a quantum adversary is attacking honest players that use quantum computing or even quantum communication themselves, 
it does not seem justified to assume that he can only communicate classically with them. Indeed, as an example, consider a zero-knowledge protocol where the prover is implemented as a small quantum device sitting inside a mobile unit, say a PDA or a smart-phone. If an adversary gets hold of the unit, he may not be able to break in and directly read the prover's secret. But he can try to subject the device to unusual physical conditions, say by cooling it down and in this way perhaps be able to communicate quantumly with the prover, even if the device was not designed for this in the first place.

In this paper, we therefore introduce a new model of quantum attacks on classical cryptographic protocols, where the adversary is allowed to ask several classical queries in quantum superposition. 
In more concrete terms, we ask, for multiparty protocols: what happens if the adversary can be in superposition of having corrupted several different subsets? or, for zero-knowledge protocols: what happens if a quantum verifier can be in superposition of having issued several different challenges to the prover? As we argued above, we believe such superposition attacks to be a valid physical concern, but they also form a very natural generalization from a theory point of view: in the literature on black-box quantum computing, quantum black-box access to a function is usually defined by extending classical black-box access such that queries are allowed to contain several inputs in superposition.
Our superposition attacks extend classical attacks in the same way.

Superposition attacks are strictly stronger than the standard one, and we consider the security of several primitives in this model:
We show that a secret-sharing scheme that is perfectly secure with threshold $t$ in the standard model is perfectly secure against superposition attacks if and only if the adversary's superposition is constrained to contain subsets of size at most $t/2$. If this condition is not satisfied, not only does perfect security fail, we show examples where the adversary may even learn the secret with certainty.

We use the secret-sharing result to construct zero-knowledge proofs for all of NP in the common reference string (CRS) model. While our protocol is classical, it is sound against a cheating unbounded quantum prover and computational zero-knowledge even if the verifier is allowed a superposition attack.
Since we use the CRS model, the reader may ask why we do not use existing protocols for non-interactive zero-knowledge (NIZK), where the prover just sends a single message to the verifier. In this way, the adversary would not get a chance to do a superposition attack. However, the most general assumption under which NIZK is known to be possible is existence of one-way permutations. They in turn are only known to be realizable under assumptions that are easily broken by a quantum adversary, such as factoring or discrete log. Therefore we do not consider NIZK a satisfactory solution. 

Finally, we consider multiparty computation and we define a UC-style model for static and passive superposition attacks on classical MPC protocols. Given our result on secret-sharing schemes, it is natural to speculate that classical MPC protocols that are secure against $t$ corruptions, are secure against superposition attacks corrupting $t/2$ players. The situation turns out to be more complicated, however:
We show that for the model that gives the adversary the most power (and hence is the most hostile to the simulator), simulation based security is not possible at all. The adversary can put its query in a state that prevents the simulator from learning any information on the inputs and outputs of corrupted players. 
However, putting a natural constraint on the adversary, we show a non-trivial example of a protocol that can indeed be simulated. By non-trivial, we mean that
although the protocol is secure against a classical attack, we can show that it cannot be proved
secure against a superposition attack by simply running the classical simulator in superposition. We therefore come up with techniques that are  ``more quantum'' to do the simulation. We give a (in completely classical terms) a characterization of the protocols that can be simulated using these techniques. The obtained simulators are not necessarily efficient, however.

Whether more general positive results hold in this constrained model remains an open question.
Likewise, the very natural question of security of {\em quantum} protocols against superposition 
attacks remains open. Note that in existing work on quantum multiparty computation~\cite{BCGHS06}, the adversary's choice of subset to corrupt is classical. The negative part of our result on secret sharing described above shows that the protocol from~\cite{BCGHS06} is not secure against superpositions attacks
as it stands.

\section{Preliminaries}
\subsection{Notation and terminology}
We will model players in protocols in two different ways: when we are not interested in computational  limitations on parties, a player will be specified by a series of unitary transforms where the $i$'th transform is done on all qubits available to the party, after the i'th message has been received (in the form of a quantum register), and then some designated part of the storage is sent as the next outgoing message. We are limiting ourselves to perfect unitary transformation of the party's register because we are exactly considering the situation where an attacker manages to prevent coupling between the party and the environment. 

In cases where we want to bound the computational complexity of a player, we consider
a players to be an infinite family of interactive quantum circuits, as in the model from \cite{FS09}, and then the complexity is circuit size.

\subsection{Running functions in superposition}

Consider any function, $f: X \rightarrow Y$ and a register of qubits, $\ket\psi = \displaystyle\sum_x \alpha_x\ket{x}\ket{0} \in \mathcal{H}_{X}\otimes \mathcal{H}_Y$, where $dim(\mathcal{H}_{X}) = |X|$ and $dim(\mathcal{H}_{Y}) = |Y|$. \em Running \em $f$ on $\ket\psi$ means to apply the unitary transformation, $U_f$, such that $ U_f \displaystyle\sum_x \alpha_x\ket{x}\ket{0} = \displaystyle\sum_x \alpha_x\ket{x}\ket{f(x)}$. In general the register in $\mathcal{H}_{Y}$, called the \em response register\em, can contain any superposition of values, not just 0. In this case, we have that, $ U_f \displaystyle\sum_{x,a} \alpha_{x,a}\ket{x}\ket{a} = \displaystyle\sum_x \alpha_{x,a}\ket{x}\ket{f(x) \oplus a}$ where $\oplus$ is the bitwise xor. 

\section{Secret sharing}
In (classical) secret sharing $n$ parties are sharing some secret value $s\in \mathbb S$ using randomness $r\in \mathcal{R}$, where $\mathbb S$ and $\mathcal{R}$ is the set of possible secrets and randomness. We name the parties $P_1, \ldots, P_n$. Let $[n] = \{ 1, \ldots, n \}$.  Each party, $P_i$, receives a {\em share}  $v_i(s, r) \in \{ 0,1\}^k$, also called his \em private view \em. That is, $v_i: \mathbb{S} \times \mathcal{R} \rightarrow \{0,1 \}^k$. For $A\subset [n]$, let $v_{A}(s,r) = \{ v_i(s, r) \}_{i \in A}$  be the string containing the concatenation of views for parties $P_i$ with $i \in A$. For convenience in the following we assume that each such string is padded, so that they have the same length regardless of the size of $A$. That is, $v_A: \mathbb{S} \times \mathcal{R} \rightarrow \{0,1 \}^t$. 
An {\em adversary structure} $G$ is a family of subsets $G \subset 2^[n]$. A secret sharing scheme
is perfectly secure against classical $G$-attacks if
for any $A\in G$,  the distribution of $v_{A}(s,r)$ does not depend on $s$. The adversary structure
of the secret sharing scheme is the maximal adversary structure $F$ such that the scheme is perfectly secure against $F$ attacks.

We'll model any passive attack on the scheme as an one-time query to an {\em corruption oracle}. The corruption oracle for a specific run of a secret sharing scheme $O(s,r,A) = v_A(s,r)$ is the function that on input $A$ returns the private view of those parties. That is, $O : \mathbb S \otimes \mathcal{R} \otimes F \rightarrow \{ 0,1 \}^t$. \\

\subsection{Two-party bit sharing example} \label{sec:2party_exp}
Before we give the full model for secret sharing we start with a small example. We consider the case of 2 parties sharing a single bit, $b\in \{ 0,1\}$ using a random bit $r\in \{0,1\}$. $[n] = \{ 1, 2\}$, $F = \{ (1), (2)\}$, $v_1(b,r) = b \oplus r$, $v_2(b,r) = r$. This scheme is trivially secure in the classical setting. In the follow we'll consider what happens if we allow the adversary to interact with the corruption oracle in superposition. As this is meant to introduce the concept, we'll cut a few corners in terms of technicality and and reserve that for the later sections.\\
Assuming some specific bit has been shared with some specific randomness, we can write the state of the parties as $\ket{v_1(b,r)} \in \mathcal{H}_2$ and $\ket{v_2(b,r)}  \in \mathcal{H}_2$. Consider an adversary supplying the following input to the corruption oracle,
\begin{equation*}
\ket \omega = \frac{1}{\sqrt{2}}(\ket{1}\ket{0} + \ket{2}\ket{0}).
\end{equation*}
The oracle will run on both of these input in superposition. The state the adversary receives will be a mixed state over different choices of randomness and secrets. We'll assume both of these are uniformly chosen and he'll hence receive,
\begin{equation*}
\rho^{adv} =\displaystyle\sum_{b\in \{0,1\},r \in \{ 0,1 \}} \frac{1}{4} \proj{\psi^{adv}_{b,r}}
\end{equation*} 
where $\ket{\psi^{adv}_{b,r}} =\frac{1}{\sqrt{2}}(\ket{1}\ket{ v_1(b,r)} + \ket{2}\ket{v_2(b,r)}) = \frac{1}{\sqrt{2}}(\ket{1}\ket{ b\oplus r} + \ket{2}\ket{r})$.
Define the state the adversary sees for a specific secret, $b$, as $\rho^{adv}_b = \displaystyle\sum_{r\in \{0,1\}} \frac{1}{2} \proj{\psi^{adv}_{b,r}}$. We would consider our bit sharing scheme secure iff for all possible queries, $\rho^{adv}_0 = \rho^{adv}_1$. However, note that,
\begin{eqnarray*}
\rho^{adv}_b =  \displaystyle\sum_{r\in \{0,1\}} \frac{1}{2} \proj{\psi^{adv}_{b,r}} =\frac{1}{2}\displaystyle\sum_{A,A' \in \{1,2\}} \ket{A}\bra{A'}\otimes\left( \displaystyle\sum_{r\in \{0,1\}}  \ket{v_A(b,r)}\bra{v_{A'}(b,r)}\right).
\end{eqnarray*}
For $A = 1, A' =2$, consider the submatrix,  $ \displaystyle\sum_{r\in \{0,1\}}  \ket{v_1(b,r)}\bra{v_{2}(b,r)} =  \displaystyle\sum_{r\in \{0,1\}} \ket{b\oplus r} \bra{b}$. It should be clear that $\displaystyle\sum_{r\in \{0,1\}} \ket{0\oplus r} \bra{0} \neq \displaystyle\sum_{r\in \{0,1\}} \ket{1\oplus r} \bra{1}$ and hence $\rho^{adv}_0 \neq \rho^{adv}_1$, which means the scheme is not secure if we allow the adversary to run the corruption oracle in superposition. This is not surprising since we know that using the Deutch-Josza algorithm you can actually distinguish between two such shares perfectly. Note that to do it perfectly it would require the use of a superposition of values for the response registers.

\subsection{Model for secret sharing}
We'll now give the full technical description of the model for superposition attacks on general secret sharing. To do this we first consider the state spaces needed to run the protocol and the attack on the protocol. First is the space that contains the shares for all the parties, $\mathcal H_{parties}$. The state in the register for this space is unchanged throughout the attack and is,
\begin{equation*}
\ket{parties}_p = \displaystyle\sum_{s \in \mathbb{S},r \in \mathcal{R}} \sqrt{p_s} \sqrt{p_r} \ket{s,r, v_{[n]}(s,r)}_p = \displaystyle\sum_{s \in \mathbb{S},r \in \mathcal{R}} \sqrt{p_s} \sqrt{p_r} \ket{s,r}\displaystyle\bigotimes_i \ket{v_i(s,r)} \in \mathcal{H}_{parties}
\end{equation*}
where $\ket{s,r}$ is the purification of the secret and randomness choice. This is purely for technical reasons and does not matter for the adversary
 as he never sees it (and hence they might as well be considered measured). Next is the space for the environment, $\mathcal{H}_{env}$, which the adversary
 can use to choose his query and use as potential auxiliary register. The initial state for the environment is a general (pure) state,
\begin{equation*}
\ket{\psi}_e = \displaystyle\sum_{x}\alpha_x\ket{x}_e \in \mathcal{H}_{env}.
\end{equation*}
Finally is the space holding the adversary's query to the corruption oracle, $\mathcal{H}_{query}$. This is initially a 'blank' state, 
\begin{equation*}
\ket{\omega}_q = \ket{0,0}_q \in \mathcal{H}_{query}.
\end{equation*}
The space for the entire model is hence, $\mathcal{H}_{total} = \mathcal{H}_{parties} \otimes \mathcal{H}_{env} \otimes \mathcal{H}_{query}$ and the intial state is,
\begin{equation*}
\ket{init}_t = \displaystyle\sum_{s \in \mathbb{S},r \in \mathcal{R}} \sqrt{p_s} \sqrt{p_r}\ket{s,r, v_{[n]}(s,r)}_p \otimes \displaystyle\sum_{x}\alpha_x\ket{x}_e\otimes\ket{0,0}_q \in \mathcal{H}_{total}.
\end{equation*}
The attack will be defined by two operations and an adversary structure $F$. First the adversary needs to construct his query for the oracle. This includes choosing the superposition of subsets he'll corrupt and associated values for the response registers. This is an arbitrary unitary operation. We'll denote it, $U^{adv,F}_{query}$,
\begin{equation*}
U^{adv,F}_{query} :  \mathcal{H}_{env} \otimes \mathcal{H}_{query} \rightarrow \mathcal{H}_{env} \otimes \mathcal{H}_{query}.
\end{equation*}
After this unitary operation the state is,
\begin{equation*}
\ket{query}_t = U^{adv,F}_{query} \ket{init}_t = \displaystyle\sum_{s \in \mathbb{S},r \in \mathcal{R}} \sqrt{p_s} \sqrt{p_r}\ket{s,r, v_{[n]}(s,r)}_p \otimes \displaystyle\sum_{x,A\in F,a \in \{ 0,1 \}^t}\alpha_{x,A,a}\ket{x}_e\otimes\ket{A,a}_q \in \mathcal{H}_{total}
\end{equation*}
where we assume it is the identity on $\mathcal{H}_{parties}$
Next the oracle, $O(s,r,A)$, is run. Let $U_{O}$ denote the unitary applying this function. The state afterwards is,
\begin{eqnarray*}
\ket{final}_t = U_{O} \ket{query}_t = \\ \displaystyle\sum_{s \in \mathbb{S},r \in \mathcal{R}} \sqrt{p_s} \sqrt{p_r}\ket{s,r, v_{[n]}(s,r)}_p \otimes \displaystyle\sum_{x,A\in F,a \in \{ 0,1 \}^t}\alpha_{x,A,a}\ket{x}_e\otimes\ket{A,a + v_A(s,r)}_q \in \mathcal{H}_{total}
\end{eqnarray*}
were we, again, assume $ U_{O}$ is padded with appropriate identities. 
Consider the final state the adversary sees for a specific secret, $s$, 
\begin{equation*}
\rho^{adv,F}_s = \displaystyle\sum_{r \in \mathcal{R}} \proj{\psi^{adv,F}_r}
\end{equation*}
where $\ket {\psi^{adv,F}_r} = \displaystyle\sum_{x,A\in F,a \in \{ 0,1 \}^t}\alpha_{x,A,a}\ket{x}_e\otimes\ket{A,a + v_A(s,r)}_q$.
\begin{definition}
A secret sharing scheme \cal S is \em perfectly secure against superposition $F$-attacks \em if, and only if, for all unitary matrices, $U^{adv,F}_{query}:  \mathcal{H}_{env} \otimes \mathcal{H}_{query} \rightarrow \mathcal{H}_{env} \otimes \mathcal{H}_{query}$ and all possible pairs of inputs, $ s,s'\in \mathbb S$, 
\begin{equation*}
\rho^{adv,F}_s = \rho^{adv,F}_{s'}
\end{equation*}
\end{definition}

For an adversary structure $F$, we define $F^2= \{ A| \ A = B\cup C\mbox{ where }B,C\in F\}$.

\begin{theorem}\label{thm:secureSS}
Let $G$ be the classical adversary structure for $\cal S$.
$\cal S$ is perfectly secure against superposition $F$-attacks if
and only if $F^2\subseteq G$.
\end{theorem}
\begin{proof}
For the forward direction, consider the adversary's final state, 
\begin{eqnarray*}
\rho^{adv}_s &=& \displaystyle\sum_{r \in \mathcal{R}} p_r\proj{\psi^{adv}_r} \\
 &=& \displaystyle\sum_{r\in\mathcal{R}, x,x', A,A\in F,a,a' \in \{ 0, 1 \}^t}p_r\alpha_{x,A,a}\alpha^ *_{x',A',a'}\ket{x}_e\bra{x}_e\otimes\ket{A,a + v_A(s,r)}_q\bra{A',a' + v_{A'}(s,r)}_q
\end{eqnarray*}
Now, for any fixed $A$, $A'$, $a$, $a'$ and $s$, consider the matrix $\displaystyle\sum_{r \in \mathcal{R}}p_r\ket{A,a +  v_A(s, r)}_q \bra{A',a' + v_{A'}(s,r)}_q$. The crucial
observation now is that this matrix is in 1-1 correspondence with the joint distribution of $v_A(s,r)$
and $v_{A'}(s,r)$. Namely, its entries are indexed by pairs of strings $(\alpha,\beta)$, where $\alpha$, $\beta$ are strings of the same length. And furthermore the $(\alpha,\beta)$'th entry is the probability that the events $v_A(s,r) = \alpha + a$ and $v_{A'}(s,r) = \beta + a'$ occur simultaneously.
Now, if $F^2 \subseteq G$, we have that $\cal S$ is perfectly secure against classical $F^2$-attacks. Therefore
the joint distribution of $v_A(s,r)$ and $v_{A'}(s,r)$ does not depend on $s$, consequently each matrix $\displaystyle\sum_{r \in \mathcal{R}}p_r\ket{A,a + v_A(r,s)}_q \bra{A',a' + v_{A'}(s,r)}_q$ is independent of $s$ as well. Hence $\forall s,s'\in \mathbb S: \rho^{adv,F}_s = \rho^{adv,F}_{s'}$ as required. 

For the only-if part, assume for contradiction that $F^2\not \subseteq G$, i.e., there exist
$A_0,A_1$ such that $A_0\cup A_1 \not\in G$. It follows that a secret shared using
$\cal S$ is uniquely determined from shares in $A_0\cup A_1 $.
Then consider the query $\ket{\omega_{{\alpha}}} = (\ket{A}\ket{0} + \ket{A'}\ket{0})/\sqrt{2}$.
By the same computation as above, we see that
$\rho^{adv}_{s}$ contains a submatrix of form
 $\displaystyle\sum_{r \in \mathcal{R}}\ket{a_A + v_{A}({s},{r})}\bra{a_{A'}+ v_{A'}({s},{r})}$, that corresponds
to the joint distribution of shares in $A$ and $A'$. But since the secret is uniquely
determined from these shares, it follows that this submatrix is different for different secrets,
and hence we get that there exists a measurement with non-zero bias towards the secret, and so $\cal S$ is not perfectly secure against quantum $F$-attacks. This is exactly the result we saw in the small example in Section \ref{sec:2party_exp}.
\end{proof}

\subsection{Simplified models for secret sharing}
When formalizing superposition attacks on secret sharing we need to consider if we allow the adversary to use different values for the response register. The choice provably make a
difference for the strength of the model, and both options can be
justified from a physical perspective. We therefore cover both
models. 
  When the adversary runs a classical component in superposition, then
  the reply will in general be a superposition. This opens the
  question of how the reply is delivered. In quantum information
  processing, it is customary that the result is xor'ed onto a
  response register $a$ supplied along with the input. I.e., for a
  function $f$ on is given a box which on classical input
  $\ket{x}\ket{a}$, the output is $\ket{x}\ket{a \oplus f(x)}$. This
  is convenient, as it is invertible, so the action of the box on a
  superposition is given by its actions on the classical inputs. This
  approach is reasonable in quantum information processing, as one is
  typically designing the boxes one self. If $f$ is a database one can
  simply design the quantum version to supply the output by xor'ing it
  onto a response register. Consider, however, the prover in a
  zero-knowledge proof, which is sent a challenge $e$ and then sends
  back $z(e)$. Even though the prover might be tricked into running on
  a superposition without noticing it, it does not seem reasonable
  that the prover would not notice it if we sent along a response
  register and asked her to xor her resply onto this register. In such
  a setting it seems more reasonable that the box/prover creates the
  response registers and returns them to the attacker/verifier. We
  model the setting of \emph{created response registers} by
  restricting the more general setting of \emph{supplied response
    registers} by allowing only $a=0$. In that case the response from
  the box would be $\ket{x}\ket{f(x)}$.
 
\subsection{Attacks on Secret Sharing}\label{sec:SSattack}
Even if a secret sharing scheme is not perfectly secure according the Theorem \ref{thm:secureSS} it does not tell us anything about how much information the adversary can actually gain on the secret. One might even hope this could become negligible by increasing the amount of randomness used to create the shares. However, in this section we show that, for any two-party Shamir secret sharing scheme, an attack can distinguish between to possible secrets with considerable bias. The attack works even in the restricted setting of supplied response registers.
\begin{lemma}
Consider a two-party Shamir secret sharing scheme $\cal S$. For any two  secrets $s, s'\in \mathbb{S} : s \neq s'$, there exists a query with $a=0$ that will allow an adversary to distinguish between the two with probability at least $p_{guess}$, where
\begin{equation*}
p_{guess} \geq \frac{3}{4}
\end{equation*}
\end{lemma}
\begin{proof}
The adversary will need no auxiliary register, so let his state simply be the query register. He constructs the following (pure state) query
\begin{equation*}
\ket{\omega} =\frac{1}{\sqrt{2}} (\ket{A_0,0} + \ket{A_1,0})
\end{equation*}
where $\ket{\omega} \in \mathcal{H}_{query}$. The final state the adversary sees for different secrets is then,

\begin{equation*}
\rho^{adv}_s = \displaystyle\sum_{r \in \mathcal{R}}  p_r \proj{\psi^{adv}_r}
\end{equation*}
where $\ket {\psi^{adv}_r} = \displaystyle\sum_{A\in \{A_0, A_1 \}} \frac{1}{\sqrt{2}} \ket{A,v_A(s,r)}_q$ and $r \in \mathcal{R}$.

It is well-known that the adversary's probability of distinguishing between two such states, $\rho^{s}_{adv}$ and $\rho^{s'}_{adv}$, is $p_{guess} = \frac{1}{2} + \frac{1}{4} \times |\rho^{s}_{adv}-\rho^{s'}_{adv}|_{Tr}$, where $|\dots|_{Tr}$ denotes the trace norm. Define the difference between the two states as the matrix $\Delta$. 
\begin{eqnarray*}
\Delta &=& \rho^{s}_{adv}-\rho^{s'}_{adv}\\
&=& \frac{1}{2}\displaystyle\sum_{A,A'\in \{A_0, A_1 \},r\in \mathcal{R}}p_r\ket{A,v_A(s,r)}_q \bra{A',v_{A'}(s,r)}_q -\frac{1}{2}\displaystyle\sum_{A,A'\in \{A_0, A_1 \},r} p_r\ket{A,v_A(s',r)}_q \bra{A',v_{A'}(s',r)}_q\\
&=& \frac{1}{2}\displaystyle\sum_{A,A'\in \{A_0, A_1 \}} \ket{A} \bra{A'} \left(\displaystyle\sum_{r \in \mathcal{R}} p_r\ket{v_A(s,r)}_q \bra{v_{A'}(s,r)}_q -\frac{1}{2}\displaystyle\sum_{r \in \mathcal{R}} p_r\ket{v_A(s',r)}_q \bra{v_{A'}(s',r)}_q\right)
\end{eqnarray*}
Since the state for any party individually is independent of the secret we have that for all $s,s'\in \mathbb{S}$

\begin{equation*}\sum_{r \in \mathcal{R}} \ket{v_A(s,r)}\bra{v_A(s,r)} = \sum_{r \in \mathcal{R}} \ket{v_A(s',r)}\bra{v_A(s',r)}
\end{equation*}
which means we only need to consider $A,A'\in \{ A_0,A_1 \}|A \neq A'$.
\begin{eqnarray*}
\Delta&=&\frac{1}{2}  \ket{A_0}\bra{A_1} \otimes\left(\displaystyle\sum_{r \in \mathcal{R}} p_r\ket{v_{A_0}(s,r)}_q \bra{{v_{A_1}}(s,r)}_q -\frac{1}{2}\displaystyle\sum_{r \in \mathcal{R}} p_r\ket{v_{A_0}(s',r)}_q \bra{v_{A_1}(s',r)}_q\right)\\
         &+& \frac{1}{2} \ket{A_1}\bra{A_0} \otimes\left(\displaystyle\sum_{r \in \mathcal{R}} p_r\ket{v_{A_1}(s,r)}_q \bra{{v_{A_0}}(s,r)}_q -\frac{1}{2}\displaystyle\sum_{r \in \mathcal{R}} p_r\ket{v_{A_1}(s',r)}_q \bra{v_{A_0}(s',r)}_q\right)\\
\end{eqnarray*}

Define the two submatrices:
\begin{equation}\label{eq:submatrix}
S =   \displaystyle\sum_{r \in \mathcal{R}} p_r\ket{v_{A_0}(s,r)}_q \bra{{v_{A_1}}(s,r)}_q -\frac{1}{2}\displaystyle\sum_{r \in \mathcal{R}} p_r\ket{v_{A_0}(s',r)}_q \bra{v_{A_1}(s',r)}_q
\end{equation}
\begin{equation*}
S^\dagger = 
 \displaystyle\sum_{r \in \mathcal{R}} p_r\ket{v_{A_1}(s,r)}_q \bra{{v_{A_0}}(s,r)}_q -\frac{1}{2}\displaystyle\sum_{r \in \mathcal{R}} p_r\ket{v_{A_1}(s',r)}_q \bra{v_{A_0}(s',r)}_q
\end{equation*}
such that $\Delta$ is the $2\times2^t$ by $2\times2^t$ matrix
\begin{equation}
\Delta =  \label{eq:delta_simplified}
\frac{1}{2}  \begin{pmatrix}
  0 & S\\
  S^\dagger & 0 
\end{pmatrix}
\end{equation}
It is well-known that if $\Delta$ is of the form (\ref{eq:delta_simplified}) and $\frac{1}{2}S$ has singular values $s_1 \geq ... \geq s_p$ then $\Delta$ has eigenvalues $\pm s_1, ..., \pm s_p$. Since $\Delta$ is Hermitian, the trace norm is the sum of the absolute eigenvalues. From this we conclude that $|\Delta|_{Tr} = |S|_{Tr}$ and we can reduce our problem to that of finding the trace norm of S. Let
\begin{eqnarray*}
S &=&  M_{s} - M_{s'} \\
M_{s} &=& \sum_{r \in \mathcal{R}}p_r\ket{v_{A_0}(s,r)}\bra{v_{A_1}(s,r)}
\end{eqnarray*}
Note that for the state to be normalized it must be that 
\begin{equation*}
\displaystyle\sum_{i,j \in \{ 0,1\}^t}[M_{s}]_{i,j} =\displaystyle\sum_{i,j \in \{ 0,1\}^t,r} p_r\delta_{v_{A_0}(s,r),i}\delta_{v_{A_1}(s,r),j} = 1
\end{equation*}
Now, define the matrix, $\tilde{M}_{s}$
\begin{equation*}
\tilde{M}_{s} = \sum_{r \in \mathcal{R}}\ket{v_{A_0}(s,r)}\bra{v_{A_1}(s,r)} + \displaystyle\sum^{t^2-1}_{i=|\mathcal{R}|}\proj{i}	
\end{equation*}
It's straight forward to see that $\tilde{M}_{s} \tilde{M}_{s}^T = \id$,
\begin{eqnarray*}
\left(\sum_{r \in \mathcal{R}}\ket{v_{A_0}(s,r)}\bra{v_{A_1}(s,r)} + \displaystyle\sum^{t^2-1}_{i=|\mathcal{R}|}\proj{i}\right) \times \left( \sum_{r' \in \mathcal{R}}\ket{v_{A_1}(s,r')} \bra{v_{A_0}(s,r')}+ \displaystyle\sum^{t^2-1}_{i=|\mathcal{R}|}\proj{i}	\right) \\
=\displaystyle\sum_{r,r' \in \mathcal{R}}\ket{v_{A_0}(s,r)} \bra{v_{A_0}(s,r')} \bra{v_{A_1}(s,r)}\ket{v_{A_1}(s,r')}+ \displaystyle\sum^{t^2-1}_{i=|\mathcal{R}|}\proj{i} \\
=\displaystyle\sum_{r\in \mathcal{R}}\ket{v_{A_0}(s,r)} \bra{v_{A_0}(s,r)} + \displaystyle\sum^{t^2-1}_{i=|\mathcal{R}|}\proj{i} = \id \\
\end{eqnarray*}
Note that $ \displaystyle\sum^{t^2-1}_{i=|\mathcal{R}|}\proj{i}	$ is simply used to pad the subspace to ensure that the matrix is unitary. It will not be of any importance in the following calculations and can simply be ignored.
It is well-known that
\begin{equation*}
|S|_{Tr} = \operator{max_U}\{|\operator{Tr}(SU)| \}
\end{equation*}
where U is any unitary matrix. In other words, any specific matrix $U$ is going to give a lower bound on the trace norm. Both $\tilde{M}_{s}$ and $\tilde{M}^T_{s}$ are such unitary matrices,
\begin{eqnarray*}
|S|_{Tr} &=& \operator{max_U}\{|\operator{Tr}(SU)| \} \geq |\operator{Tr}(S\tilde{M}^{T}_{s})|\\
&=&|\operator{Tr}((M_{s} - M_{s'})\tilde{M}^{T}_{s})| \\
&=&|\displaystyle\sum_{i,j \in \{ 0,1\}^t} [M_{s}]_{i,j}[\tilde{M}_{s}]_{i,j}- \displaystyle\sum_{i,j \in \{ 0,1\}^t} [M_{s'}]_{i,j}[\tilde{M}_{s}]_{i,j}| \\
&=&|1- \displaystyle\sum_{i,j \in \{ 0,1\}^t} [M_{s'}]_{i,j}[\tilde{M}_{s}]_{i,j}| \\
\end{eqnarray*}
Now note that 
\begin{equation*}
\displaystyle\sum_{i,j \in \{ 0,1\}^t} [M_{s'}]_{i,j}[\tilde{M}_{s}]_{i,j} = \displaystyle\sum_{i,j \in \{ 0,1\}^t,r,r'} p_r\delta_{v_{A_0}(s,r),i}\delta_{v_{A_1}(s,r),j} \delta_{v_{A_0}(s',r'),i}\delta_{v_{A_1}(s',r'),j} 
\end{equation*}
However the pair  $(v_{A_0}(s,r), v_{A_1}(s,r))$, uniquely defines $s$  and hence the sum is $0$ unless $s = s'$. Therefore for $s\neq s': |S|_{Tr} \geq 1$.
\begin{equation*}
p_{guess} = \frac{1}{2} + \frac{1}{4} |\Delta|_{Tr}  =  \frac{1}{2} + \frac{1}{4} |S|_{Tr} \geq \frac{1}{2} + \frac{1}{4}\times 1 = \frac{3}{4}
\end{equation*}
which completes the proof. \qed
\end{proof}

\newcommand{\commit}{{\tt Commit}}
\newcommand{\pk}{{\tt pk}}
\newcommand{\sk}{{\tt sk}}
\newcommand{\pkH}{{\tt pkH}}
\newcommand{\pkB}{{\tt pkB}}
\newcommand{\GH}{{\cal G}_{\tt H}}
\newcommand{\GB}{{\cal G}_{\tt B}}

\section{Zero-Knowledge}
In this section, we present a zero-knowledge proof for any NP problem in the common reference string model. The proof is sound for an unbounded prover (quantum or not) and is computationally zero-knowledge for a polynomially bounded quantum verifier, even if superposition attacks are allowed.

For the protocol, we need a commitment scheme with special properties: we require a {\em keyed} commitment scheme $\commit_\pk$, where the corresponding public key $\pk$ is generated by one of two possible key-generation algorithms: $\GH$ or $\GB$. For a key $\pkH$ generated by $\GH$, the commitment scheme $\commit_\pkH$ is unconditionally hiding, whereas the other generator, $\GB$, actually produces a key {\em pair} $(\pkB,\sk)$, so that the secret key $\sk$ allows to efficiently extract 
$m$ from $\commit_\pkB(m,r)$, and as such $\commit_\pkB$ is unconditionally binding. Furthermore, we require that keys $\pkH$ and $\pkB$ produced by the two generators
are computationally indistinguishable, for any family of polynomial size quantum circuits. 
We call such a commitment scheme a {\em dual-mode} commitment
scheme.
\footnote{The notions of dual-mode {\em cryptosystems} and of meaningful/meaningless encryptions, as introduced in~\cite{PVW08} and~\cite{KN08}, are similar in spirit but differ slightly technically. }
As a candidate for implementing such a system, we propose the public-key encryption scheme of Regev~\cite{Regev05}, which is
based on a worst-case lattice assumption and is not known to be breakable even by
(efficient) quantum algorithms. Regev does not explicitly state that the
scheme has the property we need, but this is implicit in his proof that the underlying
computational assumption implies semantic security.
\footnote{The proof
   compares the case where the public key is generated normally to a
   case where it is chosen with no relation to any secret key. It is
   then argued that the assumption implies that the two cases are
   computationally  indistinguishable, and that in the second case, a
   ciphertext carries essentially no information about the
   message. This argument implies what we need.}
   
\subsection{The Model}   
We now describe the framework for our protocol: the proof system is specified w.r.t. a language $L$, and we have a prover $P$ and a verifier $V$, both 
are assumed classical (when playing honestly). They get as input a common reference string $CRS$ chosen with a prescribed distribution $\sigma$ and a string $x$. $P$ and $V$ interact and at the end $V$ outputs $accept$ or $reject$. The first two properties we require are standard:
{\em Completeness:} if $x\in L$ and $P,V$ follow the protocol, $V$ outputs $accept$ with probability 1.
{\em Soundness:} if $x\not\in L$ (but $CRS$ is chosen according to $\sigma$) then for any prover $P^*$,
$V$ outputs $accept$ with probability negligible (in the length of $x$) when interacting with $P^*$ on input $x$ and $CRS$.

For zero-knowledge, we extend the capabilities of a cheating verifier $V^*$ so it may do a
superposition attack
For simplicity, we give our definition of superposition zero-knowledge only for 3-move public coin protocols, i.e., conversations
are assumed to have the form $(a,e,z)$, where $e$ is a random challenge issued by the verifier.
It is not hard to extend the definition but the notation becomes more cumbersome.
First, $V^*$ is assumed to be a quantum machine, and the protocol is executed as follows:
$V^*$ receives $x, CRS$ and $P$'s first message $a$. Now, in stead of sending a classical challenge
$e$, $V^*$ is allowed to send a query 
$$\sum_{e,y}\alpha_{e,y} \ket{e}\ket{y}.$$ 
We assume the the prover will process the query following his normal algorithm in superposition,
so the verifier will get the same
two registers back, in state
$$\sum_{e,y}\alpha_{e,y} \ket{e}\ket{y+z(x,e,\rho)},$$
where $z(x,e,\rho)$ is $P$'s response to challenge $e$ on input  $x$ and internal 
randomness $\rho$. Finally, $V^*$ outputs 0 or 1. Let $p_{real}(x)$ be the probability that
1 is output.
We say that the proof system is {\em superposition zero-knowledge} if there exists an polynomial time
quantum machine, the simulator $S$, such that the following holds for any cheating verifier 
$V^*$ and $x\in L$: $S$ interacts with $V^*$ on input $x$, and we let $p_{sim}(x)$ be the probability
that $V^*$ outputs 1. Then $|p_{real}(x)- p_{sim}(x)|$ is negligible (in the length of $x$).

Note that, as usual in the CRS model, $S$ only gets $x$ as input and may therefore generate the reference string itself.

\subsection{The Protocol}
We now describe the basic ideas behind our protocol: we will let the CRS contain the following:
$\pkB, c= \commit_{\pkB}(0), \pkB'$, where the public keys are both generated by $\GB$. Then, using
a standard trick, we will let $P$ show that either $x\in L$ or $c$
contains a 1. Since of course the latter statement is false, $P$ still needs to convince us that 
$x\in L$. The simulator, on the other hand, can construct a reference string where 
$c$ does contain 1 and simulate by following the protocol. The CRS will look the same to the verifier
so we just need that the change of witness used is not visible in the proof, i.e., the proof
should be witness indistinguishable.
In this way, we can simulate without rewinding, and this allows $V^*$ to be quantum. 

However, standard techniques for witness indistinguishability are not sufficient to handle a superposition attack. For this, we need to be more specific about the protocol:
a first attempt (which does not give us soundness) is that $P$ will secret-share his witness $w$ (where for the honest prover, $w$ will be a witness for $x\in L$), to create shares $s_1,...,s_n$ where we assume the scheme has $t$-privacy.
Then $P$'s first message is a set of commitments 
$a= (\commit_{\pkB'}(s_1, r_1),...,  \commit_{\pkB'}(s_n, r_n))$. The verifier's challenge $e$ will point out a random subset of the commitments, of size $t/2$, and the prover opens the commitments requested.
Intuitively, this is zero-knowledge by Theorem~\ref{thm:secureSS}: since we limit the number of shares the verifier can ask for to half the threshold of the secret sharing scheme, the state $V^*$ gets back contains no information on the secret $w$. 

On the other hand, this protocol is of course not sound, the verifier cannot check that the prover commits to meaningful shares of anything. To solve this, we make use of the ``MPC in the head'' technique
from~\cite{IKOS09}: Here, we make use of an $n$-party protocol in which the witness $w$ is secret-shared among the player, and a multiparty computation is done to check whether $w$ is correct with respect to claim on the the public input, namely in our case $x\in L$ and the $c$ from the $CRS$ contains $1$. Finally all players output $accept$ or $reject$ accordingly. It is assumed that the protocol is secure against active corruption
of $t$ players where $t$ is $\Theta(n)$. We will call this protocol $\pi_{L,CRS}$ in the following.
Several examples of suitable protocols can be found in \cite{IKOS09}.
In their construction, the prover emulates an
execution of $\pi$ in his head, and we let $v_{\pi_{L,CRS}}(i,\rho)$ denote he view of virtual player $i$,
where $\rho$ is the randomness used. The prover then commits to $v_{\pi_{L,CRS}}(i,\rho)$, for $i=1...n$ and
the verifier ask the prover to open $t$ randomly chosen views that are checked for consistency and adherence to $\pi_{L,CRS}$. It is shown in \cite{IKOS09} that if no valid witness exists for the public input, then the verifier will detect an error with overwhelming probability.

Now, observe that the process of emulating $\pi$ can be thought of as a secret sharing scheme, where the prover's witness $w$ is shared and each $v_{\pi}(i,\rho)$ is a share: indeed any $t$ shares contain no information on $w$ by $t$-privacy of the protocol. Therefore combining this with our rudimentary idea
from before gives us the solution.

\medskip

\noindent
{\bf Superposition-secure zero-knowledge proof for any $NP$-language} $L$.

\noindent
The public input is $x$, of length $k$ bits.
The distribution $\sigma$ generates the common reference string as
$\pkB, c= \commit_{\pkB}(0), \pkB'$, where the public keys are both generated by $\GB$
on input $1^k$.

\begin{enumerate}
\item
The prover $P$ emulates $\pi_{L,CRS}$ to generate $v_{\pi_{L,CRS}}(i,\rho)$ and sends
$\commit_{\pkB'}(v_{\pi_{L,CRS}}(i,\rho), r_i)$,  for $i=1...n$,  to the verifier $V$.
\item
$V$ sends a challenge $e$ designating a random subset of the commitments of size $t/2$.
\item
$P$ opens the commitments designated by $e$, $V$ checks the opened views according
to the algorithm described in \cite{IKOS09}, and accepts or rejects according to the result.
\end{enumerate}

\begin{theorem}
If $(\GB, \GH, \commit)$ form a secure dual-mode commitment scheme, then the above
protocol is complete, sound and superposition zero-knowledge.
\end{theorem}
\begin{proof}
Completeness is trivial by inspection of the protocol. Soundness follows immediately from the soundness proof in \cite{IKOS09}, we just have to observe that the fact that the prover opens
$t/2$ and not $t$ views makes no difference, in fact the proof holds as long as $\Theta(n)$ views are opened.
For zero-knowledge, we describe a simulator $S$:
It will generate a common reference string as $\pkH, c= \commit_{\pkH}(1), \pkH'$ where both
public keys are generated by $\GH$ on inout $1^k$. It then plays the protocol with $V^*$, answering its quantum
queries by following the protocol. This is possible since $c$ now contains a 1, so $S$ knows a valid
witness. To show that $V^*$ cannot distinguish simulation from protocol, we define series of games
\begin{description}
\item[Game 0] The protocol as described above, but where $P$ talks to $V^*$ doing a superposition attack.
\item[Game 1] As Game 0, but the CRS is generated as $\pkH, c= \commit_{\pkH}(0), \pkB'$
where $\pkH$ is generated by $\GH$ and $\pkB'$ is generated by $\GB$.
\item[Game 2] As Game 1, but the CRS is generated as $\pkH, c= \commit_{\pkH}(1), \pkH'$ where both
public keys are generated by $\GH$.
\item[Game 3] As Game 3, but the $P$ uses as witness the fact that $c$ contains a 1.
\end{description}
Now, Game 0 and Game 1 are computationally indistinguishable by assumption on the dual-mode commitment scheme, and the same is true for Game 1 and Game 2. Game 2 and Game 3 are statistically indistinguishable by Theorem~\ref{thm:secureSS} and the fact that commitments done using $\pkH'$ are statistically hiding. Finally, note that Game 3 is exactly the same game as the simulation.
\end{proof}

\section{Multiparty computation}
In this section we consider the models for MPC protocols. A classical passive attack on a multiparty computation
protocol looks a lot like the attacks on secret sharing: you query for
a subset and get back the party's entire view of the protocol.  Of
course, you can generalize this to a quantum attack in the same way.
And you can ask if there is some adversary structure for which the
protocol would be secure against such an attack, assuming classical
security.

 Security for MPC protocols is usually defined as an adversary's ability to distinguish between an attack in the {\em real world} where he's allowed access to a corruption oracle of some subset of the parties and an {\em ideal world} where the attack is {\em simulated} towards the adversary using the {\em ideal functionality} of the protocol. We hence need to describe both models for the real and for the ideal world. They'll need different spaces and different operations to execute. 
 As before, we have $n$ parties running the protocol. We name the parties $P_1, \ldots, P_n$. Let $[n] = \{ 1, \ldots, n \}$.  An
adversary structure is $F \subset 2^[n]$.  Each party, $P_i$, has local input, $s_i \in \mathbb{S}_i$ with is supplied by the adversary and chooses randomness, $r_i\in \mathcal{R}_i$. $s \in \mathbb{S}$ and $r\in \mathcal{R}$ denotes the concatenation of each of these. $v_i(s, r)$ and $o_i(s, r)$ is the private view and output for the party $P_i$, when the protocol has been run on inputs $s$ and using randomness $r$. Note these are functions and not general quantum operations as the parties, even the corrupted ones, are expected to run the protocol honestly. For $A\subset [n]$, let $v_{A}(s,r) = \{ v_i(s, r) \}_{i \in A}$, $s_{A} = \{ s_i \}_{i \in A}$ $o_{A}(s,r) = \{ o_i(s, r) \}_{i \in A}$ and be strings containing the concatenation of views, input and output for parties $P_i$ with $i \in A$. For convenience in the following we assume that each such string is padded, so that they have the same length (t bits) regardless
of the size of $A$. \\
\subsection{MPC model in the 'Real world'}
First we will consider the case of running and attacking the protocol in the real world. As earlier, all actions taken by the parties and the adversary will be considered purified so the overall state remains pure throughout. Consider the following space,
\begin{equation*}
\mathcal{H}_{total} = \mathcal{H}_{parties} \otimes \mathcal{H}_{in} \otimes \mathcal{H}_{out} \otimes \mathcal{H}_{env} \otimes \mathcal{H}_{query}
\end{equation*}
$\mathcal{H}_{parties}$ contains the private views and purification of the randomness for the parties. \\$\mathcal{H}_{in}$ contains the input the parties will use to run the protocol. \\$\mathcal{H}_{out}$ is where the output will be stored after running the protocol.  \\$\mathcal{H}_{env}$ is the environment and is used to store auxiliary input and any auxiliary register needed by the adversary. The dimension is therefore arbitrary, though finite. \\$\mathcal{H}_{query}$ is where the query to, and response from, the oracle will be stored. \\
Each of these subspaces are, of course, of appropriate (and finite) dimension.\\

We will break the superposition run, and attack, of a MPC protocol in the real world down into four unitaries. After each unitary we will consider the change to the description of the state. In the beginning all the registers are blank (ie. have value 0), except the environment, which might contain some (purified) auxiliary input for the adversary. The initial state is hence,
\begin{equation*}
\ket{init}^{rw}_{t} =  \displaystyle\sum_{x} \alpha_{x}\displaystyle \ket{0}_p\ket{0}_{i}\ket{0}_{o} \ket{x}_e \ket{0}_q
\end{equation*}
where $\ket{init}^{rw}_{t} \in \mathcal{H}_{total}$, $\ket{0}_{p} \in \mathcal{H}_{parties}$, $\ket{0}_{i} \in \mathcal{H}_{in}$, $\ket{0}_{o} \in \mathcal{H}_{out}$, $\ket{0}_{e} \in \mathcal{H}_{env}$ and $\ket{0}_{q} \in \mathcal{H}_{query}$, as should be expected from the notation. The superscript, $rw$, specifies that it's in the real world.\\
The first unitary is applied by the adversary and supplies the inputs to the parties. This is an arbitrary unitary operation. We'll denote it, $U^{adv}_{in}$,
\begin{equation*}
U^{adv}_{in} :  \mathcal{H}_{in} \otimes \mathcal{H}_{env} \rightarrow \mathcal{H}_{in} \otimes \mathcal{H}_{env}.
\end{equation*}
The result of this is that the input registers are now filled. These are in superposition over all possible inputs. The state after the first unitary is therefore,
\begin{equation*}
\ket{1}^{rw}_{t} =  \displaystyle\sum_{x,s} \alpha_{x,s}\displaystyle \ket{0}_p\ket{s}_{i}\ket{0}_{o} \ket{x}_e \ket{0}_q
\end{equation*}
The protocol is now run honestly without intervention from the adversary. This is a classical function run in superposition of the possible inputs and produces a corresponding superposition of private views and outputs for the parties. We'll denote this unitary, $U^{pro}_{run}$,
\begin{equation*}
U^{pro}_{run} :  \mathcal{H}_{parties} \otimes \mathcal{H}_{in} \otimes \mathcal{H}_{out} \rightarrow \mathcal{H}_{parties} \otimes \mathcal{H}_{in} \otimes \mathcal{H}_{out}
\end{equation*}
Recall that we are purifying all actions, hence also the choice of randomness when the protocol is run. The state after the second unitary is therefore,
\begin{equation}
\ket{2}^{rw}_{t} =  \displaystyle\sum_{x,s,r} \alpha_{x,s}\sqrt{p_{r}} \ket{v_{[n]}(s,r)}_p \ket{s}_{i}\ket{o_{[n]}(s,r)}_{o} \ket{x}_e \ket{0}_q \label{eq:rw2}
\end{equation}\\ 
Next the adversary needs to construct his query to the oracle. This includes choosing the superposition of subsets he'll corrupt and associated values for the response registers for input, output and view. For simplicity we'll sometimes refer to these three values by $a$. That is, $a = (a_i, a_o, a_v)$. This is an arbitrary unitary operation. We'll denote it, $U^{adv, F}_{query}$,
\begin{equation*}
U^{adv, F}_{query} :  \mathcal{H}_{env} \otimes \mathcal{H}_{query} \rightarrow \mathcal{H}_{env} \otimes \mathcal{H}_{query}
\end{equation*}
The state is now,
\begin{equation*}
\ket{3}^{rw}_{t} =  \displaystyle\sum_{x,s,r,A,a} \alpha_{x,s,A,a}\sqrt{p_{r}} \ket{v_{[n]}(s,r)}_p \ket{s}_{i}\ket{o_{[n]}(s,r)}_{o} \ket{x}_e \ket{A,a}_q
\end{equation*}\\ 
Next unitary is applied by the oracle, that, for each corrupted subset in the query, fills the input, output and view into the response register supplied by the adversary.  This is a classical function on each corrupted subset and view in the superposition which fills in the input, output and view into the response register. We'll denote this unitary, $U^{oracle}_{res}$,
\begin{equation*}
U^{oracle}_{res} :  \mathcal{H}_{parties} \otimes \mathcal{H}_{in} \otimes \mathcal{H}_{out} \otimes \mathcal{H}_{query} \rightarrow \mathcal{H}_{parties} \otimes \mathcal{H}_{in} \otimes \mathcal{H}_{out} \otimes \mathcal{H}_{query}
\end{equation*}
and the state after the fourth unitary is therefore
\begin{equation*}
\ket{4}^{rw}_{t} =  \displaystyle\sum_{x,s,r,A,a} \alpha_{x,s,A,a}\sqrt{p_{r}}\displaystyle \ket{v_{[n]}(s,r)}_p \ket{s}_{i}\ket{o_{[n]}(s,r)}_{o} \ket{x}_e \ket{A,a_i + s_A,a_o + o_A(s,r),a_v + v_A(s, r)}_q
\end{equation*}\\ 
The adversary receives the response register and must now guess if he's in the real or ideal world. He can do this using the input register, his auxiliary register and the query register. 
To see the adversary's final state we need to trace out the register holding the view of all the parties,
\begin{eqnarray}
\rho^{rw}_{adv} = Tr_p(\proj{4}^{rw}_t) &=&  \displaystyle\sum_{{r, r', s, s'}} \sqrt{p_{r}}\sqrt{p_{r'}} \ket{\psi^{adv}_{r, s}}\bra{\psi^{adv}_{r', s'}}\operator{Tr}(\ket{v_{[n]}(s,r)}_p\bra{v_{[n]}(s',r')}_p \otimes \ket{o_{[n]}(s,r)}_{o} \bra{o_{[n]}(s',r')}_{o}) \\
&=& \displaystyle\sum_{{r, r', s, s'}} \sqrt{p_{r}}\sqrt{p_{r'}}  \ket{\psi^{adv}_{r, s}}\bra{\psi^{adv}_{r', s'}}\bra{v_{[n]}(s,r)}_p\ket{v_{[n]}(s',r')}_p \times \bra{o_{[n]}(s,r)}_{o} \ket{o_{[n]}(s',r')}_{o}) \\
&=& \displaystyle\sum_{{r, s}} p_{r} \proj{\psi^{adv}_{r, s}}
\end{eqnarray} \label{eq:adv_rw_mpc_final}
where $\ket {\psi^{adv}_{r, s}} = \displaystyle\sum_{x, A,a}\alpha_{x,s,A,a}\ket{s}_{i}\ket{x}_e\ket{A,a_i + s_A,a_o + o_A(s,r),a_v + v_A(s, r)}_q$. It is interesting to note that even though the input register was supplied by the adversary, as the parties run the protocol their private state becomes entangled with the input register. This is a register the adversary does not have access to and as a consequence he now sees a mixed state over possible inputs. \\\\

We'll sum up these steps below in Figure \ref{fig:realworld}.
\begin{figure}[htp]
\begin{framed}
Inital state:
\begin{equation*}
\ket{init}^{rw}_{t} =  \displaystyle\sum_{x} \alpha_{x}\displaystyle \ket{0}_p\ket{0}_{i}\ket{0}_{o} \ket{x}_e \ket{0}_q
\end{equation*}
\begin{enumerate}
\item  
\begin{equation*}
U^{adv}_{in} :  \mathcal{H}_{in} \otimes \mathcal{H}_{env} \rightarrow \mathcal{H}_{in} \otimes \mathcal{H}_{env}
\end{equation*}
\begin{equation*}
\ket{1}^{rw}_{t} =  \displaystyle\sum_{x,s} \alpha_{x,s}\displaystyle \ket{0}_p\ket{s}_{i}\ket{0}_{o} \ket{x}_e \ket{0}_q
\end{equation*}\\
\item 
\begin{equation*}
U^{pro}_{run} :  \mathcal{H}_{parties} \otimes \mathcal{H}_{in} \otimes \mathcal{H}_{out} \rightarrow \mathcal{H}_{parties} \otimes \mathcal{H}_{in} \otimes \mathcal{H}_{out}
\end{equation*}
\begin{equation*}
\ket{2}^{rw}_{t} =  \displaystyle\sum_{x,s,r} \alpha_{x,s}\sqrt{p_{r}} \ket{v_{[n]}(s,r)}_p \ket{s}_{i}\ket{o_{[n]}(s,r)}_{o} \ket{x}_e \ket{0}_q
\end{equation*}\\
\item  
\begin{equation*}
U^{adv, F}_{query} :  \mathcal{H}_{env} \otimes \mathcal{H}_{query} \rightarrow \mathcal{H}_{env} \otimes \mathcal{H}_{query}
\end{equation*}
\begin{equation*}
\ket{3}^{rw}_{t} =  \displaystyle\sum_{x,s,r,A,a} \alpha_{x,s,A,a}\sqrt{p_{r}} \ket{v_{[n]}(s,r)}_p \ket{s}_{i}\ket{o_{[n]}(s,r)}_{o} \ket{x}_e \ket{A,a}_q
\end{equation*}\\
\item  
\begin{equation*}
U^{oracle}_{res} :  \mathcal{H}_{parties} \otimes \mathcal{H}_{in} \otimes \mathcal{H}_{out} \otimes \mathcal{H}_{query} \rightarrow \mathcal{H}_{parties} \otimes \mathcal{H}_{in} \otimes \mathcal{H}_{out} \otimes \mathcal{H}_{query}
\end{equation*}
\begin{equation*}
\ket{4}^{rw}_{t} =  \displaystyle\sum_{x,s,r,A,a} \alpha_{x,s,A,a}\sqrt{p_{r}} \ket{v_{[n]}(s,r)}_p \ket{s}_{i}\ket{o_{[n]}(s,r)}_{o} \ket{x}_e \ket{A,a_i + s_A,a_o + o_A(s,r),a_v + v_A(s, r)}_q
\end{equation*}\\
\end{enumerate}
\end{framed}
\small
\caption{Purified run of multiparty computation in the real world with superposition attacks}
 \label{fig:realworld}
\end{figure}

\newpage
\subsection{MPC model in the 'Ideal world'}
Secondly we consider the ideal world, where a simulator is to simulate a real attack for the adversary using only the ideal functionality. The space for this model is,
\begin{equation*}
\mathcal{H}_{total} = \mathcal{H}_{idealF} \otimes \mathcal{H}_{in} \otimes \mathcal{H}_{out} \otimes \mathcal{H}_{sim}   \otimes \mathcal{H}_{env} \otimes \mathcal{H}_{query}
\end{equation*}
$\mathcal{H}_{in}, \mathcal{H}_{out}$ and $\mathcal{H}_{env}$ serve the same purpose as in the real world. \\$\mathcal{H}_{sim}$ denotes the subspace in which the simulator operates. \\$\mathcal{H}_{query}$ is still the register where the adversary constructs his query, but the simulator will be allowed to change this before using it to query the ideal functionality. Finally, we have a space for the ideal functionality, $\mathcal{H}_{idealF}$. This is necessary as we want to purify the random choices made and we need the state of the ideal functionality to be entangled with the input register as the parties would be in the real world. We'll describe the unitaries that differ and refer to the earlier description for those that don't. There will be six unitaries in total. The initial state is,
\begin{equation*}
\ket{init}^{iw}_{t} =  \displaystyle\sum_{x} \alpha_{x}\displaystyle \ket{0}_{if} \displaystyle \ket{0}_i\ket{0}_{o}\ket{0}_{s} \ket{x}_e \ket{0}_q
\end{equation*}
where $\ket{init}_{t} \in \mathcal{H}_{total}$, $\ket{0}_{if} \in \mathcal{H}_{idealF}$, $\ket{0}_{i} \in \mathcal{H}_{in}$, $\ket{0}_{o} \in \mathcal{H}_{out}$, $\ket{0}_{s} \in \mathcal{H}_{sim}$, $\ket{0}_{e} \in \mathcal{H}_{env}$ and $\ket{0}_{q} \in \mathcal{H}_{query}$, as should be expected from the notation. The superscript, $iw$, denotes that it's in the ideal world.
The first unitary is exactly as in the real world and hence,
\begin{equation*}
U^{adv}_{in} :  \mathcal{H}_{in} \otimes \mathcal{H}_{env} \rightarrow \mathcal{H}_{in} \otimes \mathcal{H}_{env}
\end{equation*}
\begin{equation*}
\ket{1}^{iw}_{t} =  \displaystyle\sum_{x,s} \alpha_{x,s} \ket{0}_{if}\ket{s}_{i}\ket{0}_{o}  \ket{0}_s \ket{x}_e \ket{0}_q
\end{equation*}\\
The ideal functionality now entangles itself with the input register and produces the correct outputs in the output register. This is a classical function of each view in superposition. We'll denote this unitary, $U^{ideal}_{out}$,
\begin{equation*}
U^{ideal}_{out} :  \mathcal{H}_{idealF} \otimes \mathcal{H}_{in} \otimes \mathcal{H}_{out} \rightarrow  \mathcal{H}_{idealF} \otimes\mathcal{H}_{in} \otimes \mathcal{H}_{out}
\end{equation*}
Recall that any random choices are purified onto $\mathcal{H}_{idealF}$. The resulting state is therefore,
\begin{equation*}
\ket{2}^{iw}_{t} =  \displaystyle\sum_{x,s,r} \alpha_{x,s} \sqrt{p_{r}}\ket{s,r}_{if}\ket{s}_{i}\ket{o_{[n]}(s,r)}_{o}  \ket{0}_s \ket{x}_e \ket{0}_q
\end{equation*}\\
The adversary constructs his query exactly as earlier,
\begin{equation*}
U^{adv, F}_{query} :  \mathcal{H}_{env} \otimes \mathcal{H}_{query} \rightarrow \mathcal{H}_{env} \otimes \mathcal{H}_{query}
\end{equation*}
\begin{equation*}
\ket{3}^{iw}_{t} =  \displaystyle\sum_{x,s,r,A,a} \alpha_{x,s,A,a}\sqrt{p_{r}} \ket{s,r}_{if} \ket{s}_{i}\ket{o_{[n]}(s,r)}_{o} \ket{0}_s \ket{x}_e \ket{A,a}_q
\end{equation*}\\
Now the simulator gets the query and must construct an appropriate response to the adversary, only using the ideal functionality. First the simulator is allowed to change the adversary's query using an auxiliary register. Only requirement is that it is corruption preserving \footnote{That is, the probability of measuring a specific $A$ is unchanged}. This is an arbitrary unitary operation. We'll denote it, $U^{sim}_{query}$,
\begin{equation*}
U^{sim}_{query} :  \mathcal{H}_{sim}  \otimes \mathcal{H}_{query} \rightarrow \mathcal{H}_{sim}  \otimes \mathcal{H}_{query}
\end{equation*}
\begin{equation*}
\ket{4}^{iw}_{t} =  \displaystyle\sum_{x,z,s,r,A,a} \alpha_{x,z,s,A,a}\sqrt{p_{r}} \ket{s,r}_{if} \ket{s}_{i}\ket{o_{[n]}(s,r)}_{o} \ket{z}_s\ket{x}_e \ket{A,a}_q
\end{equation*}\\
Now the ideal functionality is run on the query from the simulator. This is a classical function that each corrupted subset and each view in the superposition which fills in the input and output into the response register. We'll denote this unitary, $U^{ideal}_{res}$,
\begin{equation*}
U^{ideal}_{res} :  \mathcal{H}_{in}  \otimes  \mathcal{H}_{out}  \otimes \mathcal{H}_{query} \rightarrow  \mathcal{H}_{in}  \otimes  \mathcal{H}_{out}  \otimes \mathcal{H}_{query}
\end{equation*}
and the response register now contains input and output from the corrupted parties, but, of course, not their views.
\begin{equation*}
\ket{5}^{iw}_{t} =  \displaystyle\sum_{x,z,s,r,A,a} \alpha_{x,z,s,A,a}\sqrt{p_{r}} \ket{s,r}_{if} \ket{s}_{i}\ket{o_{[n]}(s,r)}_{o} \ket{z}_s\ket{x}_e \ket{A,a_i + s_A,a_o + o_A(s,r),a_v}_q
\end{equation*}\\
Next the simulator must try to simulate the response the adversary got in the real world. We'll denote this unitary, $U^{sim}_{res}$,
\begin{equation*}
U^{sim}_{res} :  \mathcal{H}_{sim}  \otimes \mathcal{H}_{query} \rightarrow \mathcal{H}_{sim}  \otimes \mathcal{H}_{query}.
\end{equation*}
The resulting state is
\begin{equation*}
\ket{6}^{iw}_{t} =  \displaystyle\sum_{x,z,s,r,A,a} \alpha'_{x,z,s,A,a}\sqrt{p_{r}} \ket{s,r}_{if} \ket{s}_{i}\ket{o_{[n]}(s,r)}_{o} \ket{z}_s\ket{x}_e \ket{A,\psi_{z, s, A, a}}_q.
\end{equation*}\\	
Finally, to see what state the adversary sees we must trace out the ideal functionality, the output and the simulator registers,
\begin{eqnarray*}
\rho^{iw}_{adv} = Tr_{if,o,s}(\proj{6}^{iw}_t) &=& \displaystyle\sum_{{r},r', s, s',z,z'} p_{r} \ket{\psi^{adv}_{r,s, z}}\bra{\psi^{adv}_{r', s',z'}} tr\left(\ket{s,r}_{if}\bra{s',r'}_{if} \otimes \ket{o_{[n]}(s,r)}_{o} \bra{o_{[n]}(s',r')}_{o} \otimes \ket{z}_s \bra{z'}_s\right)\\
&=& \displaystyle\sum_{{r},r', s, s',z,z'} p_{r} \ket{\psi^{adv}_{r,s, z}}\bra{\psi^{adv}_{r', s',z'}} \left(\bra{s,r}_{if} \ket{s',r'}_{if} \times \bra{o_{[n]}(s,r)}_{o} \ket{o_{[n]}(s',r')}_{o} \times \bra{z}_s \ket{z'}_s\right)\\
&=& \displaystyle\sum_{{r}, s,z} p_{r} \ket{\psi^{adv}_{r,s, z}}\bra{\psi^{adv}_{r, s,z}}
\end{eqnarray*} \label{eq:adv_iw_mpc_final}
where $\ket {\psi^{adv}_{r,s,z}} = \displaystyle\sum_{x,A,a}\alpha_{x,z,s,A,a}\ket{s}_{i}\ket{x}_e\ket{A,\psi_{z, s, A, a}}_q$.
Again, we will sum up the steps below in Figure \ref{fig:idealworld}. 
\begin{figure}[htp]
\begin{framed}
Inital state:
\begin{equation*}
\ket{init}^{iw}_{t} =  \displaystyle\sum_{x} \alpha_{x}\displaystyle \ket{0}_{if} \displaystyle \ket{0}_i\ket{0}_{o}\ket{0}_{s} \ket{x}_e \ket{0}_q
\end{equation*}

\begin{enumerate}
\item  
\begin{equation*}
U^{adv}_{in} :  \mathcal{H}_{in} \otimes \mathcal{H}_{env} \rightarrow \mathcal{H}_{in} \otimes \mathcal{H}_{env}
\end{equation*}
\begin{equation*}
\ket{1}^{iw}_{t} =  \displaystyle\sum_{x,s} \alpha_{x,s} \ket{0}_{if}\ket{s}_{i}\ket{0}_{o}  \ket{0}_s \ket{x}_e \ket{0}_q
\end{equation*}\\
\item 
\begin{equation*}
U^{ideal}_{out} :  \mathcal{H}_{idealF} \otimes \mathcal{H}_{in} \otimes \mathcal{H}_{out} \rightarrow  \mathcal{H}_{idealF} \otimes\mathcal{H}_{in} \otimes \mathcal{H}_{out}
\end{equation*}
\begin{equation*}
\ket{2}^{iw}_{t} =  \displaystyle\sum_{x,s,r} \alpha_{x,s}\sqrt{p_{r}} \ket{s,r}_{if}\ket{s}_{i}\ket{o_{[n]}(s,r)}_{o}  \ket{0}_s \ket{x}_e \ket{0}_q
\end{equation*}\\
\item  
\begin{equation*}
U^{adv, F}_{query} :  \mathcal{H}_{env} \otimes \mathcal{H}_{query} \rightarrow \mathcal{H}_{env} \otimes \mathcal{H}_{query}
\end{equation*}
\begin{equation*}
\ket{3}^{iw}_{t} =  \displaystyle\sum_{x,s,r,A,a} \alpha_{x,s,A,a}\sqrt{p_{r}} \ket{s,r}_{if} \ket{s}_{i}\ket{o_{[n]}(s,r)}_{o} \ket{0}_s \ket{x}_e \ket{A,a}_q
\end{equation*}\\
\item  
\begin{equation*}
U^{sim}_{query} :  \mathcal{H}_{sim}  \otimes \mathcal{H}_{query} \rightarrow \mathcal{H}_{sim}  \otimes \mathcal{H}_{query}
\end{equation*}
\begin{equation*}
\ket{4}^{iw}_{t} =  \displaystyle\sum_{x,z,s,r,A,a} \alpha_{x,z,s,A,a}\sqrt{p_{r}} \ket{s,r}_{if} \ket{s}_{i}\ket{o_{[n]}(s,r)}_{o} \ket{z}_s\ket{x}_e \ket{A,a}_q
\end{equation*}\\
\item  
\begin{equation*}
U^{ideal}_{res} :  \mathcal{H}_{in}  \otimes  \mathcal{H}_{out}  \otimes \mathcal{H}_{query} \rightarrow  \mathcal{H}_{in}  \otimes  \mathcal{H}_{out}  \otimes \mathcal{H}_{query}
\end{equation*}
\begin{equation*}
\ket{5}^{iw}_{t} =  \displaystyle\sum_{x,z,s,r,A,a} \alpha_{x,z,s,A,a}\sqrt{p_{r}} \ket{s,r}_{if} \ket{s}_{i}\ket{o_{[n]}(s,r)}_{o} \ket{z}_s\ket{x}_e \ket{A,a_i + s_A,a_o + o_A(s,r),a_v}_q
\end{equation*}\\
\item 
\begin{equation*}
U^{sim}_{res} :  \mathcal{H}_{sim}  \otimes \mathcal{H}_{query} \rightarrow \mathcal{H}_{sim}  \otimes \mathcal{H}_{query}
\end{equation*}
\begin{equation*}
\ket{6}^{iw}_{t} =  \displaystyle\sum_{x,z,s,r,A,a} \alpha'_{x,z,s,A,a}\sqrt{p_{r}} \ket{s,r}_{if} \ket{s}_{i}\ket{o_{[n]}(s,r)}_{o} \ket{z}_s\ket{x}_e \ket{A,\psi_{z, s, A, a}}_q
\end{equation*}\\	
\end{enumerate}
\end{framed}
\small
\caption{Purified run of multiparty computation in the ideal world with superposition attacks}
\label{fig:idealworld}
\end{figure}

A MPC protocol is defined by the operator $U^{pro}_{run}$ (which in turn also defines $U^{ideal}_{out}$). A specific adversay is defined by the initial state $\displaystyle\sum_x \alpha_a \ket{x}$ and the two unitary operators, $U^{adv}_{in}, U^{adv, F}_{query}$. Finally, a simulator is defined by the two unitary operators, $\{ U^{sim}_{query}, U^{sim}_{res}\}$. 
This allows for the following definition.
\begin{definition}
A pair of unitary operators, $\{ U^{sim}_{query}, U^{sim}_{res}\}$, are a perfect black-box simulator for the MPC protocol defined by $U^{pro}_{run}$ with adversary structure, F, if, and only if, for all auxiliary inputs, $\displaystyle\sum_x \alpha_x \ket{x}$, and all operators $U^{adv}_{in}, U^{adv, F}_{query}$,
\begin{equation*}
 \rho^{rw}_{adv} = Tr_{p,o}(\proj{4}^{rw}_t) = \rho^{iw}_{adv} = Tr_{if,o,s}(\proj{6}^{iw}_t)
\end{equation*}
where 
\begin{eqnarray*}
\ket{4}^{rw}_t &=& U^{oracle}_{res} U^{adv, F}_{query} U^{pro}_{run} U^{adv}_{in} \ket{init}^{rw}_{t} \\
\ket{6}^{iw}_t &=& U^{sim}_{res} U^{ideal}_{res}U^{sim}_{query}U^{adv, F}_{query}U^{ideal}_{out}U^{adv}_{in} \ket{init}^{iw}_{t}
\end{eqnarray*}

where we for readability assume that the operators are padded with appropriate identities on the subspaces they don't operate.
\end{definition}
\subsection{No simulator for general MPC if $U^{sim}_{query}$ is unitary}\label{subsec:NoUniSim}
Consider a simple MPC protocol with a single dealer that deals a secret $s\in \{ 0,1\}$ to a number of parties.  Denote this dealer as the party, $P_1$. The adversary will not need an auxiliary register in the following and will therefore be left out of the equations. Let the adversary make a query, defined by $U^{adv, F}_{query} : \mathcal{H}_{query} \rightarrow \mathcal{H}_{query}$, that is only of the dealer and one more party where he puts the response register for the input in perfect superposition. The complete state after applying $U^{adv, F}_{query}$ on the query register is,
\begin{equation*}
\ket{3}^{rw}_{t} =  \displaystyle\sum_{r,a_i}\frac{1}{\sqrt{|\mathbb S|}}\sqrt{p_{r}} \ket{v_{[n]}(s,r)}_p \ket{s}_{i}\ket{o_{[n]}(s,r)}_{o}  \ket{A,a_i,0,0}_q
\end{equation*}\\ 
where $A = \{ 1,2\}$
The state after applying the oracle is
\begin{equation*}
\ket{4}^{rw}_{t} =  \displaystyle\sum_{r,a_i} \frac{1}{\sqrt{|\mathbb S|}}\sqrt{p_{r}} \ket{v_{[n]}(s,r)}_p \ket{s}_{i}\ket{o_{[n]}(s,r)}_{o} \ket{A,a_i + s_A,0,v_A(s,r)}_q
\end{equation*}\\ 
If we look at the final part of the query register ($v_A(s,r)$) we can see that it is in a classical state and contains the view of the dealer and one other party, that is, the randomness and one secret share. This uniquely defines the secret and hence the query register must be orthogonal for two different secrets.\\ 
In the ideal world assume, for the time being, that $U^{sim}_{query}$ is the identity. The state after the oracle for the ideal functionality is applied is then,
\begin{equation*}
\ket{5}^{iw}_{t} =  \displaystyle\sum_{r,a_i} \frac{1}{\sqrt{|\mathbb S|}}\sqrt{p_{r}} \ket{s,r}_{if} \ket{s}_{i}\ket{o_{[n]}(s,r)}_{o}\ket{0}_s  \ket{A,a_i + s_A,0,0}_q
\end{equation*}
Since this is in perfect superposition over all values of $a_i$ we can conclude that,
\begin{eqnarray*}
\ket{5}^{iw}_{t} &=&  \displaystyle\sum_{r,a_i} \frac{1}{\sqrt{|\mathbb S|}}\sqrt{p_{r}} \ket{s,r}_{if} \ket{s}_{i}\ket{o_{[n]}(s,r)}_{o}\ket{0}_s \ket{A,a_i + s_A,0,0}_q \\
&=&  \displaystyle\sum_{r,a_i} \frac{1}{\sqrt{|\mathbb S|}}\sqrt{p_{r}} \ket{s,r}_{if} \ket{s}_{i}\ket{o_{[n]}(s,r)}_{o}\ket{0}_s  \ket{A,a_i,0,0}_q
\end{eqnarray*}
and the state the simulator sees is therefore independent of the secret. There's hence no way it can produce two orthogonal states depending on the secret and hence no way to simulate. For the case of a different simulator where $U^{sim}_{query}$ is not the identity, but some unitary operation on $\mathcal{H}_{query}$, there exists a different adversary that applies the unitary, $\tilde{U}^{adv}_{query} = U^{adv, F}_{query}(U^{sim}_{query})^{-1}$ instead. The register the simulator sends to the oracle of the ideal functionality is now exactly the same as above, and the same argumentation can now be repeated. It follows that for any such simulator there exists an adversary that cannot be simulated.
  \subsection{No simulator for quantum attacks by running classical simulator in superposition} \label{subsec:NoClasSup}

Unfortunately, it is not completely clear how to design a good quantum
simulator, even assuming restrictions on $F$ as in Theorem \ref{thm:secureSS} for secret sharing and in the setting with created response registers. A natural first
attempt would be to use a classical simulator (which we can assume
exists) and produce a superposition of what it produces on those
corrupted subsets that occur in the query. We can show, however, that
this cannot work in general.

Let us first make clear what we mean by running a classical simulator
in superposition: consider a classical machine $S$ which gets as input
parties subset $A$, the inputs and outputs of those parties $(s_A,o_A(s))$ and a random string $c$. It then
outputs $S(A,s_A, o_A(s),c)$. Running $S$ in superposition now means that on
input $\ket {\psi^{sim}_{x,s}} = \displaystyle\sum_{ A \in F,c}\alpha_{x,s,A}\ket{0}_s\ket{A, s_A,o_A(s),0}_q$, we
output $\displaystyle\sum_{ A\in F}\alpha_{x,s,A}\sqrt{p_c}\ket{c}_s\ket{A, s_A,o_A(s),S(A,s_A,o_A(s),c)}_q$.  This means that the
state returned to the adversary will be
$$
\sum_{A,A'\in F}\alpha_{A}\alpha^*_{A'}\ket{A}\bra{A'}\otimes\ket{s_A, o_A(s)}\bra{s_{A'}, o_{A'}(s)}\otimes
         \sum_c p_c\ket{S(A,s_A, o_A(s),c)}\bra{S(A',s_{A'}, o_{A'}(s),c)}
$$

Now consider the following simple example protocol: we have 4 parties
$P_0, P_1,P_2, P_3$. Player $P_0$ gets as input a bit $s$. He will
then secret share it additively among the other parties: he chooses
bits $r_1, r_2$ at random and sends $r_1$ to $P_1$, $r_2$ to $P_2$ and
$s\oplus r_1 \oplus r_2$ to $P_4$. There is no output defined for
anyone.  Clearly, this protocol is perfectly secure against a
classical attack where at most 2 parties are corrupted. One might
therefore hope that it would be perfectly secure against quantum
attacks using superpositions of sets containing only 1 party.

However, we now argue that such security cannot be shown by running 
any classical simulator $S$ in superposition.
To this end, consider the case where we corrupt all parties in equally weighted superposition.
This will mean that the simulator has to start from the input state. 
$$\frac{1}{2}( \ket{P_0}\ket{s}\ket{0} + \sum_{i=1}^3\ket{P_i}\ket{\bot}\ket{0})$$
The state has this form since the input and output for $P_0$ is just $s$
and the other parties have no input or output. The state returned will be
$$
\rho_{sim,s}= \sum_{A,A'\in \{ \{P_0\} , \{ P_1\}, \{P_2\} , \{ P_3\}  \}}\frac{1}{2}\ket{A}\bra{A'}\otimes\ket{s_A}\bra{s_{A'}}
         \otimes \sum_c p_c\ket{S(A,s_A,c)}\bra{S(A',s_{A'},c)}
$$
If we define $s_{\{ P_i\}} = s_i$ is some secret $s$ if $i=0$ and $\bot$ otherwise, and
index with $P_i,P_{i'}$ instead of $A,A'$, we can write the state a bit more 
conveniently as:
$$
\rho_{sim,s}= \sum_{i,i'}\frac{1}{2}\ket{P_i}\bra{P_{i'}}\otimes\ket{s_i}\bra{s_{i'}}
         \otimes \sum_c p_c\ket{S(i,s_i,c)}\bra{S(i',s_{i'},c)}
$$

On the other hand, we can compute the state $\rho_{real,s}$
that would be returned from a real attack. Define $v_i(s,r_1,r_2)$ to be
the view of the protocol for $P_i$, defined as a 2-bit register. Thus
$$v_0(s,r_1,r_2)= (r_1,r_2), v_1(s,r_1,r_2)= (r_1,0), v_2(s,r_1,r_2)= (r_2,0),
v_3(s,r_1,r_2)= (r_1\oplus r_2\oplus s ,0).$$
This means that the state returned for a particular choice of $s, r_1, r_2$ is
$$
\ket{\Psi_{s,r_1,r_2}} = \frac{1}{2} \sum_{i=0}^3 \ket{P_i}\ket{s_i}\ket{v_i(s,r_1,r_2)}
$$
For fixed $s$, each choice of $r_1,r_2$ occurs with probability $1/4$, so
$$
\rho_{real,s} = \sum_{r_1,r_2} \frac{1}{4} \ket{\Psi_{s,r_1,r_2}}\bra{\Psi_{s,r_1,r_2}}
= \frac{1}{16}\sum_{i,i'}\ket{P_i}\bra{P_{i'}} \ket{s_i}\bra{s_{i'}}
\sum_{r_1,r_2}\ket{v_i(s,r_1,r_2)}\bra{v_{i'}(s,r_1,r_2)}
$$
Consider the part of $\rho_{sim,s}$ that corresponds to $P_i= P_0$
and $P_{i'}= P_1$. Then, since we assume perfect simulation, i.e., $\rho_{sim,s}= \rho_{real,s}$,
for any $c$, we must have $S(P_0,s_A, o_A(s),c)= S(0, s, c) = (r_1,r_2)$
and $S(P_{i'},s_{A'}, o_{A'}(s),c)= S(1, \bot, c) = (r_1,0)$, where the {\em same} $r_1$ occurs
in both strings, since in a real execution, $P_1$ would of course receive the same bit
that $P_0$ sent. We get an exactly similar conclusion for $P_{i'}= P_2$, and for 
$P_{i'}= P_3$, we can conclude that $S(3,\bot, c)= (s\oplus r_1\oplus r_2, 0)$.

But now note that, we can simply compute
$S(i,\bot,c)$ for $i=1,2,3$ and some fixed $c$. Then the above shows
that if running $S$ in superposition was a perfect quantum simulator, we could compute
$(s,0)= S(1,\bot,c)\oplus S(2,\bot,c)\oplus S(3,\bot,c)$, without any information
on $s$, which is of course a contradiction.
\subsection{Limited simulators}
The result in \ref{subsec:NoUniSim} strongly suggests that it's impossible to construct a simulator for general MPC protocols in the setting with supplied response registers. In this section we will discuss the problem of constructing a simulator for the model with created response registers and protocols for deterministic functions. That is, the output does not depend on the chosen randomness. While simulators in general consists of two operators, $\{ U^{sim}_{query}, U^{sim}_{res}\}$, we will restrict $U^{query}_{sim}$ to be the identity. That is, the query from the adversary is sent directly to the oracle. We can do this wlog because there's no values in the response register and $U^{query}_{sim}$ must be corruption preserving.

Assume all randomness is uniformly chosen. This can be done wlog and helps to unclutter the notation. 
For a specific protocol, define the matrix (or vector of states) $M(A,s)$ as 
\begin{equation*}
M(A,s) = (\ket{A,v_A(s,0)}, \dots , \ket{A,v_A(s,r)}, \dots, \ket{A,v_A(s,|\mathcal{R}|-1)}) 
\end{equation*}
We can now express the security of the protocol in the following way
\begin{lemma}\label{lem:secureMPC}
A multiparty computation protocol for a deterministic function is perfectly secure against quantum $F$-attacks w, and only if, there exists a set of $|\mathcal{R}|\times |\mathcal{R}|$ unitary matrices $\{U_s\}_{s\in \mathbb{S}}$ such that for all $s,s'\in \mathbb{S},r\in\mathcal{R},A \in F_{s,s'}$ where $F_{s,s'} = \{ A \in F | s_A = s'_A \wedge o_A(s) = o_A(s')\}$
\begin{eqnarray*}
M(A,s) U_s = M(A,s') U_{s'}
\end{eqnarray*}
\end{lemma}

\begin{proof}
Consider the final state the adversary sees (using the output is independent of randomness and only 0 in the response registers),
\begin{equation*}
\rho^{rw}_{adv} = Tr_p(\proj{4}^{rw}_t) = \displaystyle\sum_{{r \in \mathcal{R}}} \prob \proj{\psi^{adv}_{r}}
\end{equation*} 
where $\ket {\psi^{adv}_{r}} = \displaystyle\sum_{x,s,A}\alpha_{x,s,A}\ket{s}_{i}\ket{x}_e\ket{A, s_A,o_A(s),v_A(s, r)}_q$
and the state the simulator sees after the oracle has been applied in the ideal world,
\begin{equation*}
\rho^{iw}_{sim} = Tr_{if,e,i}(\proj{5}^{rw}_t) = \displaystyle\sum_{x,s} \proj{\psi^{sim}_{x,s}}
\end{equation*} 
where $\ket {\psi^{sim}_{x,s}} = \displaystyle\sum_{\textbf A}\alpha_{x,s,A}\ket{0}_s\ket{A, s_A,o_A(s),0}_q$. The state the simulator must sent back is,
\begin{equation*}
\displaystyle\sum_{x,r,s}\prob \proj{\psi_{x,r,s}}
\end{equation*}
where
\begin{equation*}
\ket{\psi_{r,s}}= \displaystyle\sum_{A} \alpha_{x,s,A} \ket{A, s_A, o_A(s),v_A(s,r)}_q 
\end{equation*}
 To continue we need a claim that is almost equivalent to Theorem 4 in ~\cite{CJW03}, but changed slightly for our specific purposes. For this reason we'll also provide a separate proof. The reader is encouraged to read ~\cite{CJW03} for additional information on the existence of transformations between sets of quantum states.
\begin{claim} 

There exists a perfect simulator, i.e. $\rho^{rw}_{adv} = \rho^{iw}_{adv}$, iff there exist unitary operator, $U^{sim}_{res}$ such that for all $x$, $s\in \mathbb{S}$
\begin{equation}
	U^{sim}_{res} \ket {\psi^{sim}_{x,s}} = U^{sim}_{res} \displaystyle\sum_{ A}\alpha_{x,s,A}\ket{0}_s\ket{A, s_A,o_A(s),0}_q =  \rootprob\displaystyle\sum_{A}\alpha_{x,s,A}\displaystyle\sum_{i,k}[U_s]_{i,k}\ket{k}_s \ket{A, s_A, o_A(s),v_A(s,i)}_q  \label{eq:wellformed}
\end{equation}
where $[U_s]_{i,k}$ is the $i,k$ index of a unitary matrix in some set of unitary matrices, $\{U_{s}\}_{s\in \mathbb{S}}$.
\end{claim}
\begin{proof}
For the forward direction, consider the final state the adversary sees
\begin{equation*}
\rho^{iw}_{adv} = Tr_{if,s}(\proj{6}^{iw}_t) = \displaystyle\sum_{k,s}\proj{\psi^{adv}_{k,s}}
\end{equation*}
where $\ket {\psi^{adv}_{k,s}} = \displaystyle\sum_{x,A}\alpha_{x,s,A}\ket{s}_{i}\ket{x}_e \left(\displaystyle\sum_{i}[U_s]_{i,k}\rootprob \ket{A, s_A, o_A(s),v_A(s,i)}_q\right)$.
Note that $\forall s,A,A',x,x'$,
\begin{eqnarray} \label{eq:qmagic}
\prob \displaystyle\sum_{i, j,k'}[U^*_s]_{i,k}[U_{s}]_{j,k} \ket{A, s_A, o_A(s),v_A(s,i)} \bra{A', s_{A'}, o_{A'}(s),v_{A'}(s,j)} \\
=\prob  \displaystyle\sum_{i,k}[U^*_s]_{i,k}[U_s]_{i,k} \ket{A, s_A, o_A(s),v_A(s,i)} \bra{A', s_{A'}, o_{A'}(s),v_{A'}(s,i)} \nonumber \\ 
= \prob \displaystyle\sum_{i}p_i \ket{A, s_A, o_A(s),v_A(s,i)} \bra{A', s_{A'}, o_{A'}(s),v_{A'}(s,i)} \nonumber
\end{eqnarray} 
where we used that $[U_s]_{i,k}$ is unitary and only depends on s. This means the adversary sees the state he expects (except for labels) and we've completed the proof for the forward direction. \\
Conversely, assume for contradiction that there is a simulator that can constructs the correct mixed state for the adversary, but none that can construct one of the form of equation (\ref{eq:wellformed}). If $\displaystyle\sum_{x,k,s} {p_k} \proj{\phi_{x,k,s}}$ is the state the simulator sents back, then it must be that $\forall x, s:$
\begin{equation*}
\displaystyle\sum_{k} {p_k} \proj{\phi_{x,k,s}} = \displaystyle\sum_{r \in \mathcal{R}}\prob \proj{\psi_{x,r,s}}
\end{equation*}
where
\begin{equation*}
\ket{\psi_{r,s}}= \displaystyle\sum_{A} \alpha_{x,s,A} \ket{A, s_A, o_A(s),v_A(s,r)}_q 
\end{equation*}
This is true if, and only if, there exists unitary matrices, $\{U_{x,s}\}_{x,s}$, such that,
\begin{eqnarray*}
\forall k: \sqrt{p_k} \ket{\phi_{x,k,s}} = \rootprob \displaystyle\sum_{i} [U_{x,s}]_{i,k} \ket{\psi_{x,i,s}}.
\end{eqnarray*}
By purifying the state we see we get,
\begin{equation*}
\displaystyle\sum_k\sqrt{p_k}\ket{k} \ket{\phi_{x,k,s}} = \rootprob \displaystyle\sum_{i,k} [U_{x,s}]_{i,k} \ket{k} \ket{\psi_{x,i,s}}\\
= \rootprob\displaystyle\sum_{A}\alpha_{x,s,A} \displaystyle\sum_{i,k} [U_{x,s}]_{i,k}\ket{k}_s  \ket{A, s_A, o_A(s),v_A(s,i)}_q.
\end{equation*}
Note that we can choose the purification because all purifications are unitarily equivalent. And since the value of $x$ only effects the amplitude we can, by linearity, assume that $U_{x,s}$ does not depend on $x$. This shows that the state produced by any perfect simulator must be of that form. This is a contradiction and completes the proof. \qed
\end{proof}

$U^{sim}_{res}$ exists (and is unitary) if it preserves inner product between all possible states, hence we can conclude that there exists a perfect simulator if there exists a set of unitary matrices, $\{U_s\}_{s\in \mathbb{S}}$ such that for all $x,x', s, s'\in \mathbb S$ and all queries $\{ \alpha_{x,s,A} \}, \{ \alpha'_{x',s',A}\}$
\begin{eqnarray*}
( \displaystyle\sum_{A}\alpha^*_{x,s,A}\bra{0}_s\bra{A, s_A,o_A(s),0}_q) (\displaystyle\sum_{A'}\alpha'_{x',s',A'}\ket{0}_s\ket{A', s'_{A'},o_{A'}(s'),0}_q) \\
=\prob \left(\displaystyle\sum_{A}\alpha^*_{x,s,A}\displaystyle\sum_{i,k}[U^*_s]_{i,k} \bra{k}_s \bra{A, s_A, o_A(s),v_A(s,i)}_q\right)\\\left(\displaystyle\sum_{A'}\alpha'_{x',s',A'}\displaystyle\sum_{j,k'}[U_{s'}]_{j,k'} \ket{k'}_s \ket{A', s'_{A'}, o_{A'}(s'),v_{A'}(s',j)}_q\right) \\
\end{eqnarray*}
For the LHS we have that, 
\begin{eqnarray*}
( \displaystyle\sum_{A}\alpha^*_{x,s,A}\bra{0}_s\bra{A, s_A,o_A(s),0}_q) (\displaystyle\sum_{A'}\alpha'_{x',s',A'}\ket{0}_s\ket{A', s'_{A'},o_{A'}(s'),0}_q) \\
= \displaystyle\sum_{A}\alpha^*_{x,s,A}\alpha'_{x',s',A}\bra{s_A,o_A(s)}\ket{s'_A,o_A(s')} = \displaystyle\sum_{A \in F_{s,s'}}\alpha^*_{x,s,A}\alpha'_{x',s',A}
\end{eqnarray*}
For the RHS we get that,
\begin{eqnarray*}
\prob\left(\displaystyle\sum_{A}\alpha^*_{x,s,A}\displaystyle\sum_{i,k}[U^*_s]_{i,k} \bra{k}_s \bra{A, s_A, o_A(s),v_A(s,i)}_q\right)\\\left(\displaystyle\sum_{A'}\alpha'_{x',s',A'}\displaystyle\sum_{j,k'}[U_{s'}]_{j,k'} \ket{k'}_s \ket{A', s'_{A'}, o_{A'}(s'),v_{A'}(s',j)}_q\right) \\
=\prob \displaystyle\sum_{A,k,i,j}\alpha^*_{x,s,A}\alpha'_{x',s',A}[U^*_s]_{i,k}[U_{s'}]_{j,k}\bra{s_A, o_A(s),v_A(s,i)}\ket{s'_A, o_A(s'),v_A(s',j)} \\
=\prob \displaystyle\sum_{s, s', A,k}\alpha^*_{x,s,A}\alpha'_{x',s',A}\bra{s_A,o_A(s)}\ket{s'_A,o_A(s')} \times \left(\displaystyle\sum_{i} [U^*_s]_{i,k} \bra{A, v_A(s,i)}\right) \left(\displaystyle\sum_{j}[U_{s'}]_{j,k}\ket{A, v_A(s',j)}\right) 
\end{eqnarray*}	
As the state in equation (\ref{eq:wellformed}) must be normalized it follows that the maximum value of
$$\prob \left(\displaystyle\sum_{i} [U^*_s]_{i,k} \bra{A, v_A(s,i)}\right) \left(\displaystyle\sum_{j}[U_{s'}]_{j,k}\ket{A, v_A(s',j)}\right)$$ is $\prob$, so for the entire sum is only one if, and only if, $$\displaystyle\sum_{i} [U_s]_{i,k} \ket{A, v_A(s,i)} = \displaystyle\sum_{j}[U_{s'}]_{j,k}\ket{A, v_A(s',j)}.$$ Hence we see that, there exists a simulator iff there exist $|\mathcal{R}|\times |\mathcal{R}|$ unitary matrices $U_s$ such that for all $s, s'\in \mathbb{S}, k\in\mathcal{R},A \in F_{s,s'}$
\begin{eqnarray}
\displaystyle\sum_{i} [U_s]_{i,k}\ket{A, v_A(s,i)} =\displaystyle\sum_{j} [U_{s'}]_{j,k} \ket{A, v_A(s',j)} \label{eq:secMPCForEachr}
\end{eqnarray}

Using the definition of $M(A,s)$ we can now write the requirement in the more convenient way: There exists a simulator for the MPC protocol if, and only if, there exist a set of unitary matrices, $\{U_s\}_{s\in \mathbb{S}}$, such that for all $s,s'\in \mathbb{S}, A\in F_{s,s'}:$
\begin{equation*}
M(A,s) U_{s} = M(A,s') U_{s'}
\end{equation*}
which completes the proof of Lemma \ref{lem:secureMPC}. \qed
\end{proof}
We'd like to warn the reader that it is important not to confuse the unitary matrices, $\{U_s\}_{s\in \mathbb{S}}$, with the simulator itself. They are merely a tool to show the existence of one. Also, at this point the reader would be excused for lacking any intuition on why the lemma is reasonable. We therefore find it instructive to consider two examples. First we'll reconsider secret sharing in this framework and secondly a simple MPC protocol. 
\subsection{Secret sharing example}
For this example we'll show that all secret sharing schemes secure against superposition F-attacks in fact satisfies the requirement in Lemma \ref{lem:secureMPC} (as they should). This does not add to what we already knew and is only meant to illustrate the principles. We'll allow each choice of randomness to be non-uniform as it adds only very little clutter. We note that since no parties have any input or output we have that for all $s,s'\in \mathbb{S}: F_{s,s'} = F$. The result follows as a Corollary of Theorem \ref{thm:secureSS}.
\begin{corollary}\label{corol:secureSS}
If a secret sharing scheme $\cal S$ is perfectly secure against quantum $F$-attacks then there exists a set of $|\mathcal{R}|\times |\mathcal{R}|$ unitary matrices $\{U_s\}_{s\in \mathbb{S}}$ such that for all $s,s'\in \mathbb{S},r\in\mathcal{R},A\in F$
\begin{eqnarray*}
\displaystyle\sum_{i\in \mathcal{R}} \sqrt{p_i}[U_s]_{i,r} \ket{v_A(i,s)} = \displaystyle\sum_{j\in \mathcal{R}} \sqrt{p_j}[U_{s'}]_{j,r} \ket{v_A(j,s')}
\end{eqnarray*}
\end{corollary}
\begin{proof}
Since the secret sharing scheme is secure we know from Theorem \ref{thm:secureSS} that the joint distribution of the view for any $A, A' \in F$ is independent of s and hence, for all $s,s'\in \mathbb{S}$;
\begin{equation*}
\sigma_{{s} } = \displaystyle\sum_{r\in \mathcal{R}} p_r\ket{\psi_{s,r}}\bra{\psi_{s,r}} = \sigma_{{s'}}  = \displaystyle\sum_{r\in \mathcal{R}}  p_r \ket{\psi_{s',r}}\bra{\psi_{s',r}} 
\end{equation*}
where 
\begin{equation*}
\ket{\psi_{s,r}} = \displaystyle\sum_{A\in F} \alpha_{A} \ket{A} \ket{v_A(r,s)}
\end{equation*}
This is equivalent to saying that the fidelity of two such states is $1$. According to Uhlmann's Theorem this implies that there exists purifications of $\sigma_{s}$ and $\sigma_{s'}$ such that their inner product is $1$. Unitary equivalence of purifications implies you can write any such purification in $\mathcal{H}\otimes \mathcal{R}$, where $dim(\mathcal{R}) = |\mathcal{R}|$, as
\begin{equation*}
\displaystyle\sum_{k\in \mathcal{R}}\sqrt{p'_k}\ket{\psi'_{s,k}} \ket{k}_{\mathcal{R}}= \displaystyle\sum_{i,k\in \mathcal{R}}\sqrt{p_i}[U_s]_{i,k}\ket{\psi_{s,i}} \ket{k}_{\mathcal{R}} 
\end{equation*}
 where $\{U_s\}_{s\in \mathbb{S}}$ is a set of $|\mathcal{R}| \times |\mathcal{R}|$ unitary matrices. As the fidelity must be $1$, it must
 be that for all $s,s'\in \mathbb{S}$
\begin{eqnarray*}
\left(\displaystyle\sum_{i,k\in \mathcal{R}} \sqrt{p_i}[U^*_s]_{i,k} \bra{\psi_{s,i}} \bra{k}_{\mathcal{R}} \right)\left(\displaystyle\sum_{j,k'\in \mathcal{R}} \sqrt{p_{j}}[U_{s'}]_{j,k'} \ket{\psi_{s',j}} \ket{k'}_{\mathcal{R}}\right) = \displaystyle\sum_{i,j,k\in \mathcal{R}} \sqrt{p_i}\sqrt{p_j} [U^*_s]_{i,k} [U_{s'}]_{j,k}\bra{\psi_{s,i}}\ket{\psi_{s',j}}\\
=\displaystyle\sum_{A\in F,k\in \mathcal{R}} |\alpha_{A}|^2\left(\displaystyle\sum_{i\in \mathcal{R}} \sqrt{p_i} [U^*_s]_{i,k} \bra{v_A(i,s)}\right)\left(\displaystyle\sum_{j\in \mathcal{R}} \sqrt{p_j}[U_{s'}]_{j,k} \ket{v_A(j,s')} \right) =1
\end{eqnarray*}
Because $\displaystyle\sum_{A\in F} |\alpha_{A}|^2 = 1$ and noting the states are normalized, we can conclude that
\begin{equation*}
\forall s,s'\in \mathbb{S},k\in \mathcal{R},A\in F:\left(\displaystyle\sum_{i\in \mathcal{R}} \sqrt{p_i} [U^*_s]_{i,k} \bra{v_A(i,s)}\right)\left(\displaystyle\sum_{j\in \mathcal{R}} \sqrt{p_j}[U_{s'}]_{j,k} \ket{v_A(j,s')} \right) =1
\end{equation*}
From which the result follows. \qed
\end{proof}
\subsection{Simple MPC simulator example}
For this example we return to the simple four-party ($P_0,P_1,P_2,P_3)$ bit-sharing scheme considered in section \ref{subsec:NoClasSup}. There we showed that we could not simulate an attack by simply running a classical simulator for the protocol in superposition. We are now in position to show that a simulator nonetheless exists. Recall that $\ket{v_0(s,r_0,r_1)}= \ket{r_0,r_1}, \ket{v_1(s,r_0,r_1)}= \ket{r_0,0}, \ket{v_2(s,r_1,r_2)}= \ket{r_1,0},
\ket{v_3(s,r_0,r_1)}= \ket{r_0\oplus r_1\oplus s ,0}$. We have two possible inputs $s \in \{0,1\}$ so we need to find two unitary matrices, $U_{0}, U_{1}$ in order to apply Lemma \ref{lem:secureMPC}. These have been found manually. 
\begin{equation*}U_{0} =
\begin{pmatrix}
 1 & 0 & 0 & 0  \\
 0 & 1 & 0 & 0\\
 0 & 0 & 1 & 0 \\
 0 & 0 & 0 & 1
\end{pmatrix}
\end{equation*}
\begin{equation*}
U_1 =
\begin{pmatrix}
 \frac{1}{2} & \frac{1}{2} & \frac{1}{2} & -\frac{1}{2} \\
 \frac{1}{2} & \frac{1}{2} & -\frac{1}{2} & \frac{1}{2} \\
 \frac{1}{2} & -\frac{1}{2} & \frac{1}{2} & \frac{1}{2} \\
 -\frac{1}{2} & \frac{1}{2} & \frac{1}{2} & \frac{1}{2} 
\end{pmatrix}
\end{equation*}
It is a tedious, but straight forward, calculation to show that\footnote{Note that we encode the randomness as $(0,0) = 0, (0,1)  = 1, (1,0) = 2, (1,1) = 3$},
\begin{eqnarray*}
M(P_1,0) = M(P_1,1) U_{1}\\
M(P_2,0) = M(P_2,1) U_{1}\\
M(P_3,0) = M(P_3,1) U_{1}
\end{eqnarray*} 
and since $P_0 \notin F_{0,1}$ we do \em not \em require that $M(P_0,0) = M(P_0,1) U_{1}$. To get a better feeling for what is going on we'll consider equation (\ref{eq:secMPCForEachr}) for two specific choices for $A$ and $r$. First let $A = P_2$ and $r = (1,0)$. The following must be true for the unitary matrices to be correct choices,
\begin{equation*}
\ket{P_2, v_{P_2}(0,(1,0))} =\displaystyle\sum_{(i_0, i_1) \in \mathcal{R}} [U_1]_{(i_0, i_1),(1, 0)} \ket{P_2, v_{P_2}(1,(i_0, i_1))}
\end{equation*}
Since $U_{0}$ is just the identity, the LHS is easy to calculate, $\ket{P_2, v_{P_2}(0,(1,0))} = \ket{0,0}$ For the RHS we see that
\begin{eqnarray*}
\displaystyle\sum_{(i_0, i_1) \in \mathcal{R}} [U_1]_{(i_0, i_1),(1, 0)} \ket{P_2, v_{P_2}(1,(i_0, i_1))} \\= \frac{1}{2}\ket{P_2, v_{P_2}(1,(0, 0))} - \frac{1}{2}\ket{P_2, v_{P_2}(1,(0, 1))} +  \frac{1}{2}\ket{P_2, v_{P_2}(1,(1, 0))} +  \frac{1}{2}\ket{P_2, v_{P_2}(1,(1, 1))}\\
= \frac{1}{2}\ket{0,0} - \frac{1}{2}\ket{1,0} +  \frac{1}{2}\ket{0,0} +  \frac{1}{2}\ket{1,0} = \ket{0,0}
\end{eqnarray*}
Of particular interest is $A = P_3$ as the view depends on the secret (for fixed randomness). Choose $r = (0,0)$:
\begin{equation*}
\ket{P_3, v_{P_3}(0,(0,0))} =\displaystyle\sum_{(i_0, i_1) \in \mathcal{R}} [U_1]_{(i_0, i_1),(0, 0)} \ket{P_3, v_{P_3}(1,(i_0, i_1))}
\end{equation*}
For LHS, $\ket{P_3, v_{P_3}(0,(0,0))} = \ket{0,0}$. For RHS,
\begin{eqnarray*}
\displaystyle\sum_{(i_0, i_1) \in \mathcal{R}} [U_1]_{(i_0, i_1),(0, 0)} \ket{P_3, v_{P_3}(1,(i_0, i_1))} \\= \frac{1}{2}\ket{P_3, v_{P_3}(1,(0, 0))} + \frac{1}{2}\ket{P_3, v_{P_3}(1,(0, 1))} +  \frac{1}{2}\ket{P_3, v_{P_3}(1,(1, 0))} -  \frac{1}{2}\ket{P_3, v_{P_3}(1,(1, 1))}\\
= \frac{1}{2}\ket{1,0} + \frac{1}{2}\ket{0,0} +  \frac{1}{2}\ket{0,0} -  \frac{1}{2}\ket{1,0} = \ket{0,0} 
\end{eqnarray*}
\subsection{General MPC}
In this section we'll give a restatement of Lemma \ref{lem:secureMPC}, expressing the requirement for the existence of a simulator as an explicit property of the multiparty computation protocol. This will allow for a straight-forward, albeit extremely inefficient, method for checking the security of any deterministic MPC protocol in this model. For all ordered pairs of inputs, $s,s'\in \mathbb{S}$, and all sets $A\in F_{s,s'}$ we'll associate a permutation of the randomness, $\{\pi_{s,s',A}\}_{s,s',A \in F_{s,s'}}\in S(\mathcal{R})$. By ordered pairs we mean that $\pi_{s,s',A}$ may differ from $\pi_{s',s,A}$.
\begin{theorem}\label{thm:secureMPC}
A multiparty computation protocol for a deterministic function is perfectly secure against quantum $F$-attacks with created response registers if, and only if, there exist permutations, $\{\pi_{s,s',A}\}_{s,s',A\in F_{s,s'}}$ with the following two properties, 
\begin{enumerate}
\item 
\begin{eqnarray*}
\forall  s,s' \in \mathbb S, \forall A\in F_{s,s'}, \forall r \in \mathcal{R} :\\
 \ket{v_{A}(s,\pi_{s,s', A}(r))} = \ket{v_{A}(s',\pi_{s',s, A}(r))}
\end{eqnarray*}
\item 
\begin{eqnarray*}
\forall s,s',s'' \in \mathbb S, \forall A\in F_{s,s'}, A' \in F_{s,s''}: \\
\displaystyle\sum_{r\in \mathcal{R}} \ket{v_{A}(s,r)} \bra{v_{A'}(s,r)} = \displaystyle\sum_{r \in \mathcal{R}} \ket{v_{A}(s,\pi_{s,s', A}(r))} \bra{v_{A'}(s,\pi_{s,s'', A'}(r))}
\end{eqnarray*}

\end{enumerate}
Note that property (1) is exactly the statement that a (not necessarily efficient) simulator exists in the classical model.
\end{theorem}
\begin{proof}
Before we begin we'll need the following claim,
\begin{claim}
We can without loss of generality assume that all the rows (and columns) of $U_s$ sum to 1.
\end{claim}
\begin{proof}
First recall that according to Lemma \ref{lem:secureMPC}  we have that there is a simulator iff
\begin{eqnarray*} \label{eq:theunitaries}
\exists U_{0}, \cdots, U_{|\mathbb{S}| -1} \text{ st}.\\
\forall s,s' \in \mathbb{S}, \forall A\in F_{s,s'}: M(A,s)U_s &=&  M(A,s') U_{s'}
\end{eqnarray*}
For any solution we can always multiply with $U^{-1}_{0}$ on the right side of all equations and we can hence wlog assume that $U_{0} = I$. 
Now consider $\forall s \in \mathbb{S}$ the equation,
\begin{equation*}
\forall A\in F_{0,s}: M(A,0) =  M(A,s) U_{s}
\end{equation*}
That is, 
\begin{equation*}
\forall r\in \mathcal{R}: \displaystyle\sum_{k\in \mathcal{R}}[U_s]_{r,k} \ket{A v_{A}(s,k)} = \ket{v_{A}(0,r)}
\end{equation*}
Hence all rows (and columns) of $U_s$ must sum to $1$.
\qed
\end{proof}
 We can, in other words, view $M(A_{0,s},s) U_{s}$ simply as $M(A_{0,s},0)$ under some permutation of the columns. In fact, since $U_s$ must always preserve the length of the columns, any $M(A,s) U_s$ is always just a permutation of the columns in $M(A,s)$. Although the permutation is defined by the same unitary it is not necessarily the same permutation. But they're clearly related. We will make this relationship explicit by separating the requirement for a simulator into two parts. Let $M^{\pi_{s,s',A}}(A,s)$ denote the permutation of the columns in $M(A,s)$ that corresponds to applying the permutation function to the randomness for the view in each column, that is,
\begin{equation*}
M^{\pi_{s,s',A}}(A,s) = \left(\ket{A,v_A(s,\pi_{s,s',A}(0))}, \dots , \ket{A,v_A(s,\pi_{s,s',A}(r))}, \dots\right)
\end{equation*}
 The first requirement is simply the statement that for $A \in F_{s,s'}$ some permutation of the randomness exist to allow their views to be equal for $s$ and $s'$. The second is the statement that, when the input to the parties are the same, these permutations must be performed by the same unitary matrix. That is, there exists a simulator if, and only if, there exist permutations, $\{\pi_{s,s',A}\}_{s,s',A \in F_{s,s'}}$ with the following two properties, \\
\begin{enumerate}
\item 
\begin{eqnarray*}
\forall s,s' \in \mathbb{S}, \forall A\in F_{s,s'}: \\
M^{\pi_{s,s',A}}(A,s) = M^{\pi_{s',A}}(A,s')
\end{eqnarray*}
\item 
There exist unitary matrices, $\{U_s\}_{s\in \mathbb{S}}$, such that $\forall s,s', A \in F_{s,s'}: M(A,s)U_s = M^{\pi_{s,s',A}}(A,s)$. 
\end{enumerate}
 For a specific choice of permutations such unitary matrices exist iff they preserve the inner product. That is, iff
\begin{eqnarray*}
\forall s,s', s'' \in \mathbb{S}, \forall A\in F_{s,s'}, A'\in F_{s,s''}:\\
M(A,s) M(A',s)^\dagger = M^{\pi_{s,s',A}}(A,s) M^{\pi_{s,s'',A'}}(A',s)^\dagger
\end{eqnarray*}
Writing out the equations using the definition of $M(A,s)$ we conclude the proof. \qed
\end{proof}

Had the choice of corrupted parties been classical we could have let the unitary matrices depend on both the input \emph{and which party was corrupted}. Specifically, in equation (\ref{eq:qmagic}), the reason the unitaries cannot depend on A is that  $\displaystyle\sum_{i, j,k\in \mathcal{R}}[U^*_{s,A}]_{i,k}[U_{s,A'}]_{j,k}$ does not cancel out correctly if they differ and the adversary would not see the correct state. If the choice of $A$ was classical we would see no such cross-terms \footnote{Recall that there are no cross-terms for $s$ because the input register is perfectly entangled with the parties/ideal functionality} and no such relationship would be required.
\bibliographystyle{alpha}	
\bibliography{qip,crypto,procs,local}
\end{document}